\documentclass[twocolumn,aps,english,pra,nofootinbib,superscriptaddress,%
    longbibliography,floatfix]{revtex4-2}

\usepackage{graphicx,graphics,epsfig,subfigure,times,bm,bbm,amssymb,amsmath}
\usepackage{amsfonts,amsthm,mathrsfs,MnSymbol,accents}
\usepackage[matrix,frame,arrow]{xy}
\usepackage[pdftex]{color}
\usepackage{hypernat}
\usepackage{babel}
\usepackage{mathtools}
\usepackage[pdfstartview=FitH]{hyperref}
\usepackage[utf8]{inputenc}
\hypersetup{
    colorlinks=true,   
    linkcolor=cyan,         
    citecolor=magenta,      
    filecolor=magenta,      
    urlcolor=cyan,          
    runcolor=cyan,
    pdfencoding=unicode
}
\usepackage{braket}
\usepackage{enumerate}
\usepackage[capitalise]{cleveref}
\usepackage{tikz}

\newtheorem{mytheorem}{Theorem}
\newtheorem{mylemma}{Lemma}
\newtheorem{mycorollary}{Corollary}
\newtheorem{myproposition}{Proposition}
\newtheorem{myclaim}{Claim}
\newtheorem{mydefinition}{Definition}
\Crefname{myproposition}{Proposition}{Propositions}
\Crefname{mylemma}{Lemma}{Lemmas}

\newcommand{\mc}{\mathcal}

\newcommand{\bes} {\begin{subequations}}
\newcommand{\ees} {\end{subequations}}
\newcommand{\bea} {\begin{eqnarray}}
\newcommand{\eea} {\end{eqnarray}}
\newcommand{\beq}{\begin{equation}}
\newcommand{\eeq}{\end{equation}}

\newcommand{\abs}[1]{\ensuremath{\left|#1\right|}} 
\newcommand{\norm}[1]{\ensuremath{\left\|#1\right\|}} 

\newcommand{\ignore}[1]{}
\newcommand{\ignoreforclass}[1]{}

\newcommand{\mL}{\mathcal{L}}
\newcommand{\mLPT}{\mathcal{L}^{\textrm{PT}}}

\DeclareMathOperator{\rank}{rank}

\def\a{\alpha}
\def\b{\beta}
\def\g{\gamma}
\def\d{\delta}

\def\m{\mu}

\def\r{\rho}

\def\o{\omega}

\DeclareMathOperator{\Tr}{Tr}
\DeclareMathOperator{\id}{id}
\DeclareMathOperator{\real}{Re}
\DeclareMathOperator{\imag}{Im}

\def\>{\rangle}
\def\<{\langle}

\newcommand{\brak}[2]{\<{#1}|{#2}\>}
\newcommand{\ketb}[2]{|{#1}\>\!\<#2|}

\def\dgr{\dagger}

\usepackage{xcolor}

\usepackage[normalem]{ulem}

\begin{document}

\newcommand{\mycite}[1]{[\ref{bib-#1}]}

\title{Which differential equations correspond to the Lindblad equation?}

\author{Victor Kasatkin}
\email{kasatkin@usc.edu}
\affiliation{Center for Quantum Information Science \& Technology,
University of Southern California, Los Angeles, CA 90089, USA}

\author{Larry Gu}
\email{ljgu@usc.edu}
\affiliation{Department of Mathematics,
University of Southern California, Los Angeles, CA 90089, USA}

\author{Daniel A. Lidar}
\email{lidar@usc.edu}
\affiliation{Departments of Electrical Engineering, Chemistry, Physics and Astronomy, and Center for Quantum Information Science \& Technology,
University of Southern California, Los Angeles, CA 90089, USA}

\date{\today}

\begin{abstract}
The Lindblad master equation can always be transformed into a first-order linear ordinary differential equation (1ODE) for the coherence vector. We pose the inverse problem: given a finite-dimensional, non-homogeneous 1ODE, does a corresponding Lindblad equation exist? If so, what are the corresponding Hamiltonian and Lindblad operators? We provide a general solution to this problem, including a complete positivity test in terms of the parameters of the 1ODE. We also derive a host of properties relating the two representations (master equation and 1ODE) of independent interest.
\end{abstract}

\maketitle

\section{Introduction}
\label{sec:intro}

The Gorini-Kossakowski-Lindblad-Sudarshan (GKLS) master equation~\cite{Gorini:1976uq,Lindblad:76} is widely used in modeling the evolution of open quantum systems subject to Markovian dynamics~\cite{Alicki:87,Weiss:book,Breuer:book,rivas-open-2012,Lidar:2019aa}. It can be written in the following form:
\bes
\label{eq:LE}
\begin{align}
\label{eq:LEa}
\dot{\rho} &= \mc{L}\rho\ ,\quad 
\mc{L} = \mc{L}_H + \mc{L}_a\\
\label{eq:L_H}
\mc{L}_H[\cdot] &= -i[H,\cdot] \\
\mc{L}_a[\cdot] &= \sum_{i,j=1}^{d^2-1} a_{ij}\big(F_i \cdot F_j -\frac{1}{2}\{F_j F_i,\cdot\}\big)\ ,
\label{eq:L_a}
\end{align}
\ees 
commonly known simply as the Lindblad equation.\footnote{Sadly, Prof. Lindblad passed away while this work was in preparation. We dedicate this work to his memory.} Here, the dot denotes a time-derivative, $\rho$ is the state (density matrix) of the open system whose Hilbert space is $d$-dimensional, $\mc{L}$ is the Lindbladian, $H=H^\dagger$ is the system Hamiltonian, $\{F_i\}$ is a Hermitian operator basis, and the constants $a_{ij}\in\mathbb{C}$ are the elements of
a positive semidefinite matrix of rates $a$ (we provide complete definitions below). The positivity property of $a$ is crucial since it is necessary and sufficient for the superoperator $\mc{L}$ to generate a completely positive (CP) map. This is known as the GKLS theorem.\footnote{This is actually a pair of theorems that were discovered independently and nearly simultaneously for the finite-dimensional case by Gorini, Kossakowski, and Sudarshan~\cite{Gorini:1976uq}, and the infinite-dimensional case by Lindblad~\cite{Lindblad:76}; for a detailed account of this history see~\cite{Chruscinski:2017ab}.}

When $a\not\ge 0$, \cref{eq:LE} still describes a time-independent Markovian quantum
master equation but no longer generates a CP map. This situation arises, e.g.,
when the initial system-bath state is sufficiently
correlated~\cite{Pechukas:94,Alicki:1995,Pechukas:1995tb,Rodriguez:08,Buscemi:2013,Dominy:2016xy}.

The Lindblad equation is a first-order linear differential equation, and it can easily
be transformed to make this explicit. By expanding the density matrix $\rho$ in
the operator basis $\{F_i\}$ one straightforwardly derives the equivalent
nonhomogeneous form
\beq
\dot{\vec{v}} = (Q + R)\vec{v} + \vec{c}\ ,
\label{eq:2}
\eeq
where the ``coherence vector'' $\vec{v}$ collects the real-valued expansion coefficients of $\rho$, $Q=-Q^T$ is an antisymmetric matrix determined by $H$, while $R$ and $\vec{c}$ are determined by $a$~\cite{Alicki:87}.

The problem we address in this work is the inverse question: 

\emph{When does a general, nonhomogeneous linear first-order differential equation} 
\beq
\dot{\vec{v}} = G\vec{v} + \vec{c}
\label{eq:3}
\eeq 
\emph{describe a Markovian quantum master equation, and in particular a Lindblad equation?}

An elementary necessary condition is that $G$ must have eigenvalues whose real parts are non-positive (otherwise $\lim_{t\to\infty}\|\vec{v}(t)\|$ is unbounded). 
Here we go well beyond this observation and give a complete solution to the problem of
inverting both $H$ and $a$ from $G$ and $\vec{c}$. We also reformulate the condition
for complete positivity in terms of $G$ and $\vec{c}$.

Our results apply in the finite-dimensional setting, and we leave the inverse problem for the infinite-dimensional or unbounded settings open.

We start by stating our main result -- the solution of the inverse problem along with the complete positivity condition -- in \cref{sec:main-result}. We devote the rest of this work to providing a general background, a derivation and proof of the solution, and examples. In more detail, the background is given in \cref{sec:background}, where we introduce the representation of general superoperators over a ``nice'' operator basis such as the Gell-Mann matrices, use it to express the Markovian quantum master equation over real vector spaces, and derive a variety of general properties of $Q$, $R$, and $\vec{c}$. In \cref{sec:sol}, we derive the solution of the inverse problem and provide its proof. 
In \cref{sec:examples}, we illustrate a few aspects of the general theory in terms of examples. Given the question we pose in the title of this work, it is natural to ask about the probability that a randomly selected pair $(G,\vec{c})$ will give rise to a valid Lindblad equation. We provide a partial answer in \cref{sec:rarity}, where we point out that Lindbladians are extremely rare in the space of random matrices. In \cref{sec:prior} we discuss the relationship between our work and previous results.
We conclude in \cref{sec:conc} and provide additional supporting material in the appendices, including the analytical solution of \cref{eq:3} in \cref{app}.

\section{Main result}
\label{sec:main-result}
Consider a finite-dimensional Hilbert space $\mc{H}$, with $d=\dim(\mc{H})<\infty$, and the Banach space (a complete, normed vector space) $\mc{B}(\mc{H})$ of operators acting on $\mc{H}$ equipped with the Hilbert-Schmidt inner product $\langle A,B\rangle \equiv \Tr(A^\dagger B)$. Throughout this work, we use the following definitions.

\begin{mydefinition}
  \label{def:nice-operator-basis}
    A ``nice operator basis'' for $\mc{H}$ is a set $\{F_j\}_{j=0}^J\in\mc{B}(\mc{H})$, $J \equiv d^2-1$, such that 
    $F_0 = \frac{1}{\sqrt{d}}I$ ($I$ denotes the identity operator in $\mc{B}(\mc{H})$), $F_j = F_j^\dag$ and $\Tr F_j=0$ for $1\le j \le J$, and~\cite{Fano:1957wf}:
\beq
\label{eq:F-reqs}
 \<F_j, F_k\>=\Tr(F_j F_k)=\d_{jk} \ , \quad 0\le j \le J \ .
\eeq 
\end{mydefinition}

\begin{mydefinition}
Any equation of the form 
\cref{eq:LE}, where $H=H^\dag$ is a Hermitian operator in $B(\mathcal{H})$ and 
$a=a^\dag$ is a Hermitian $J\times J$ matrix, is a ``Markovian quantum master equation''. In this case, the superoperator $\mathcal{L}$ in \cref{eq:LE} is called a ``Liouvillian''.

When, in addition, $a \geq 0$, \cref{eq:LE} is a
``Lindblad master equation'', ``Lindblad equation'',
or ``completely positive Markovian quantum master equation''.
In this case, the superoperator $\mathcal{L}$ is called a ``Lindbladian'' or ``completely positive Liouvillian''.

The superoperator $\mc{L}_a$ is called the \emph{dissipator}. 

\end{mydefinition}

Whenever we refer to \cref{eq:LE} or the superoperator $\mathcal{L}$ without
explicitly specifying complete positivity, we assume the former case, i.e.,
we only assume the conditions $H = H^\dag$, $a = a^\dag$.

Our main result is the following:
\begin{mytheorem}
\label{thm:main-result}
  For any $J\times J$ matrix $G$ and any vector $\vec{c}$ of length $J$ with real
  coefficients there is a pair $(H, a)$ describing a Markovian quantum master equation 
  \cref{eq:LE} which transforms to \cref{eq:3} with these $G$ and $\vec{c}$.
  If, in addition to the above,
  we require $H$ to be traceless, then such a pair $(H, a)$ is unique and
  can be computed as follows:
  \bes
    \label{eq:GctoHa}
    \begin{align}
      \label{eq:GctoH}
      H &= \frac1{2id}\sum_{m,n=1}^{J} G_{nm} \left[F_m,F_n\right], \\
      \label{eq:Gctoa}
      a_{mn} &=
        \sum_{i=1}^{J} \Tr\bigg[\bigg(
            \sum_{j=1}^{J} G_{ij} F_{j} + c_{i} I
          \bigg)F_{m}F_{i}F_{n}\bigg].
    \end{align}
  \ees
  Moreover, the matrix $a=\{a_{mn}\}$ is positive semidefinite
  [i.e., \cref{eq:LE} is a Lindblad master equation]
  if and only if (iff) for all traceless $B\in \mc{B}(\mc{H})$:
  \begin{equation}
    \label{eq:Gc.positive}
    \sum_{i=1}^{J} \Tr\bigg[\bigg(
        \sum_{j=1}^{J} G_{ij} F_{j} + c_{i} I
      \bigg)B^\dagger F_{i}B\bigg] \geq 0\ .
  \end{equation}
\end{mytheorem}
Condition~\eqref{eq:Gc.positive} is equivalent to
the positive semidefiniteness
of $a$, which is simpler to check. Thus, in practice, it is preferable to first
compute $a$ using \cref{eq:Gctoa} and check the sign of its smallest eigenvalue, rather
than to work directly with \cref{eq:Gc.positive}.

Note that adding a constant, i.e., a term of the form $cI$ with $c\in\mathbb{R}$,
to $H$ does not change the Liouvillian $\mathcal{L}$.
Therefore, if the requirement that $H$ is traceless is not imposed, $H$ can
only be recovered from $G$ up to an additive constant.

\section{Background}
\label{sec:background}

Let $M(d, F)$ denote the vector space of $d\times d$ matrices with coefficients
in $F$, where $F \in \{\mathbb{R}, \mathbb{C}\}$. 
For our purposes it suffices to identify $\mc{B}(\mc{H})$ with $M(d,\mathbb{C})$. Quantum states are represented by density operators $\rho\in \mc{B}_+(\mc{H})$ (positive trace-class operators acting on $\mc{H}$) with unit trace: $\Tr\rho =1$.
Elements of $\mathcal{B}[\mathcal{B}(\mc{H})]$, i.e., linear transformations $\mc{E}:\mathcal{B}(\mc{H})\to \mathcal{B}(\mc{H})$, are called superoperators, or maps. 

Complete positivity of a superoperator $\mc{E}$ is equivalent to the statement that  $\mc{E}$ has a Kraus representation~\cite{Kraus:83}: $\forall X\in \mc{B}(\mc{H})$, 
\beq
\mc{E}(X) = \sum_i K_i X K_i^\dagger\ ,
\label{eq:KOSR}
\eeq
where the $\{K_i\}$ are called Kraus operators. When they satisfy $\sum_i K_i^\dagger K_i = I$, the map $\mc{E}$ is trace preserving. 

A quantum dynamical semigroup is defined as a one-parameter family of strongly continuous CP maps, satisfying $\Lambda_0=I$ (the identity operator), and the Markov property $\Lambda_s\Lambda_t = \Lambda_{s+t}$ for all $s,t\ge 0$.\footnote{Strong continuity means that $\forall X\in\mc{B}(\mc{H})$, $\|\Lambda_t X-X\|\to 0$ as $t\searrow 0$, i.e., the map $\Lambda_t$ is continuous in the strong operator topology.} 

We can more formally state the GKLS theorem(s)~\cite{Gorini:1976uq,Lindblad:76}
 as follows: a superoperator $\mc{L}:\mc{B}(\mc{H})\to \mc{B}(\mc{H})$ is the generator of a quantum dynamical semigroup iff it is of the form given in \cref{eq:LE}, with $a\ge 0$. I.e., when $\mc{L}$ is in the form of \cref{eq:LE}, the solution of $\dot{\rho} = \mc{L}\rho$ is $\rho(t) = \Lambda_t\rho(0)$, where $\Lambda_t = e^{\mc{L}t}$ is an element of a quantum dynamical semigroup.

\subsection{General superoperators over a nice operator basis}
\label{ss:gen-supop}

Given a nice operator basis, let us ``coordinatize'' the  operator $X\in\mc{B}(\mc{H})$ in the nice operator basis as
\begin{equation}
  \label{eq:coordinatization1}
  X=\sum_{j=0}^{J}\boldsymbol{X}_{j}F_{j},
\end{equation}
where we use the notation $\boldsymbol{X}=\{\boldsymbol{X}_{j}\}$ for the vector of coordinates of the operator $X$ (we interchangeably use the $\vec{X}$ and $\boldsymbol{X}$ notations to denote a vector). I.e., for any $X\in\mc{B}(\mc{H})$: 
\beq
  \label{eq:coordinatization2}
  \boldsymbol{X}_i = \<F_i,X\> = \Tr(F_i X)\ .
\eeq
In particular, applying \cref{eq:coordinatization1,eq:coordinatization2} to $\ketb{k}{l}$ we obtain
\begin{equation}
  \label{eq:coordinatization3}
  \ketb{k}{l} = \sum_{i=0}^{J} (F_{i})_{lk} F_i,
\end{equation}
or
\begin{equation}
  \label{eq:coordinatization4}
  \sum_{i=0}^{J} F_{i\,lk} F_{i\,k'l'} = \delta_{ll'}\delta_{kk'}.
\end{equation}
\begin{myproposition}
\label{prop:0}
The operator $X\in \mc{B}(\mc{H})$ is Hermitian iff $\boldsymbol{X}\in \mathbb{R}^{d^{2}}$ in a nice operator basis.
\end{myproposition}

\begin{proof}
If $X$ is Hermitian, then $\boldsymbol{X}_i^* =  \<F_i,X\>^* = \<X,F_i\> = \Tr(X F_i) = \Tr(F_i X) = \boldsymbol{X}_i$. Vice-versa, if all the coefficients $\boldsymbol{X}_i$ are real, then the operator $X = \sum_{j=0}^J \boldsymbol{X}_j F_j$ is Hermitian because it is a sum of Hermitian operators $F_j$ with real coefficients.
\end{proof}

More generally, for Hermitian matrices, the inner product becomes $\langle A,B\rangle = \Tr(AB)$. Such matrices have $d^2=J+1$ real parameters [$d$ real diagonal elements plus $(d^2-d)/2$ independent complex off-diagonal elements], so they can be represented in terms of vectors in $\mathbb{R}^{d^{2}}$. Moreover, if $A$ and $B$ are both Hermitian operators then: 
\beq
\langle A,B\rangle=\sum_{{j,k=0}}^{J}\boldsymbol{A}_{j}\boldsymbol{B}_{k}\Tr(F_{j}F_{k})=\sum_{{j,k=0}}^{J}\boldsymbol{A}_{j}\boldsymbol{B}_{k}\delta_{jk}=\boldsymbol{A}^{t}\boldsymbol{B}\ ,
\label{eq:<A,B>}
\eeq
so that in these coordinates, the inner product $\langle A,B\rangle$
corresponds to the standard inner product in $\mathbb{R}^{d^{2}}$. 

The matrix elements of a superoperator $\mc{E}: \mc{B}(\mc{H})\to \mc{B}(\mc{H})$ are given by 
\beq
\boldsymbol{\mathcal{E}}_{ij}= \<F_i,\mathcal{E}(F_j)\> = \Tr[F_{i}\mathcal{E}(F_{j})]\ .
\label{eq:Eij}
\eeq

The action of a general (not necessarily CP) superoperator $\mc{E}\in \mathcal{B}[\mathcal{B}(\mc{H})]$ can always be represented for any $A\in\mathcal{B}(\mc{H})$ as
\begin{equation}
\label{eq:Ph-gen-supop}
\mc{E}(A) = \sum_{{i,j=0}}^{J} c_{ij} F_i A F_{j} ,
\end{equation}
where $c_{ij}\in\mathbb{C}$ (see \cref{sec:supop.FAF}). Thus the matrix $c = \{c_{ij}\}$ specifies $\mc{E}$ in the given basis $\{F_i\}$. CP maps are a special case of Eq.~\eqref{eq:Ph-gen-supop}, where $c$ is positive semidefinite. 

For superoperators we denote the Hilbert-Schmidt adjoint of $\mc{E}$ by $\mc{E}^\dagger$, which is defined via
\beq
\langle \mc{E}^\dag(A),B\rangle \equiv \langle A,\mc{E}(B)\rangle\quad \forall A,B\in\mc{B}(\mc{H}) .
		\label{eq:hs-adjoint}
\eeq

\begin{mydefinition}
\label{def:Herm-supop}
$\mathcal{E}$ is \emph{Hermiticity-preserving} iff:
\beq
\mathcal{E}^\dag\left(A^{\dagger}\right) = \mathcal{E}(A)\ \quad \forall A\in \mc{B}(\mc{H}) .
\label{eq:Herm-pres}
\eeq
$\mc{E}$ is \emph{Hermitian} iff 
\beq
\mathcal{E}^\dag\left(A\right) = \mathcal{E}(A)\ \quad \forall A\in \mc{B}(\mc{H}) .
\eeq
\end{mydefinition}

Let us find explicit formulas for $[\mc{E}(A^\dag)]^\dag$ and $\mc{E}^\dagger(A)$ in terms
of the general representation of a superoperator given by \cref{eq:Ph-gen-supop}.
First, from Eq.~\eqref{eq:Ph-gen-supop} we have
\begin{equation}
   [\mc{E}(A^\dag)]^\dag
  = \left(\sum_{ij} c_{ij} F_i A^\dag F_j\right)^\dag
  = \sum_{ij} c^*_{ij} F_j A F_i .
  \label{eq:supop-adjoint0}
\end{equation}
Second, substituting \cref{eq:Ph-gen-supop} into
the complex conjugate of \cref{eq:hs-adjoint} yields:
\bes
\begin{align}
  &\left<B,\mc{E}^\dag(A)\right> = \left<\mc{E}(B), A\right>
  = \Tr\left(\sum_{ij} (c_{ij} F_i B F_j)^\dag A\right) \\
  \;&= \Tr\left(\sum_{ij} c_{ij}^* F_j B^\dag F_i A\right)
  = \Tr\left(B^\dag\sum_{ij} c_{ij}^* F_i A F_j \right) \\
  \;&= \<B, \sum_{ij} c_{ij}^* F_i A F_j \>.
\end{align}
\ees
Since this equality holds for every $B$, we have:
\beq
\label{eq:supop-adjoint1}
\mc{E}^\dag(A) = \sum_{ij} c_{ij}^* F_i A F_j .
\eeq

It will turn out to be useful to have another representation of $\mc{E}^\dag(A)$. Consider two sets of arbitrary operators $\{L_p\}$ and $\{M_q\}$, which we can expand in the nice operator basis as $L_p = \sum_{i=0}^J l_{ip}F_i$ and $M_q = \sum_{i=0}^J m_{iq}F_i$. In addition, let $b$ be some arbitrary matrix in $M(d^2,\mathbb{C})$ and let $c = l b m^\dag$. Then, using \cref{eq:Ph-gen-supop}:
\begin{subequations}
\begin{align}
\mc{E}(A) &= \sum_{i,j=0}^J [\sum_{p,q} l_{ip} b_{pq} m^*_{jq}] F_i A F_j \\
&= \sum_{p,q} b_{pq} \left(\sum_{i=0}^J l_{ip} F_i\right) A \left(\sum_{j=0}^J m_{jq}^* F_j\right) \\
&= \sum_{p,q} b_{pq} L_p A M_q^\dag\ ,
\end{align}
\end{subequations}
and, using \cref{eq:supop-adjoint1}:
\begin{subequations}
\begin{align}
\mc{E}^\dag(A) &= \sum_{i,j=0}^J [\sum_{p,q} l_{ip} b_{pq} m^*_{jq}]^* F_i A F_j \\
&= \sum_{p,q} b_{pq}^* \left(\sum_{i=0}^J l_{ip}^* F_i\right) A \left(\sum_{j=0}^J m_{jq} F_j\right) \\
&= \sum_{p,q} b_{pq}^* L_p^\dag A M_q\ .
\label{eq:supop-adjoint2}
\end{align}
\end{subequations}

\begin{myproposition}
$\mathcal{E}$ is simultaneously Hermiticity-preserving and Hermitian iff $c$ in \cref{eq:Ph-gen-supop} is real-symmetric, i.e., $c_{ij}=c_{ji}=c^*_{ij}$ $\forall i,j$. 
\end{myproposition}

\begin{proof}
This follows immediately by equating the right-hand sides of \cref{eq:supop-adjoint0,eq:supop-adjoint1}, both of which are then equal to $\mc{E}(A)$ by \cref{def:Herm-supop}.
\end{proof}

\begin{myproposition}
\label{prop:herm-preserving}
$\mathcal{E}$ is Hermiticity-preserving iff $\boldsymbol{\mathcal{E}}\in M(d^2,\mathbb{R})$.
\end{myproposition}

\begin{proof}
Using Eq.~\eqref{eq:Eij}, we obtain:
\beq
\label{eq:E-real}
\boldsymbol{\mathcal{E}}_{ij}^*=\Tr\left(\left[\mathcal{E}\left(F_{j}\right)\right]^{\dagger}F_{i}^{\dagger}\right)=\Tr\left(\mathcal{E}\left(F_{j}\right)F_{i}\right)=\boldsymbol{\mathcal{E}}_{ij}\ .
\eeq

In the other direction, assume $\boldsymbol{\mathcal{E}}\in M(d^2,\mathbb{R})$. 
$\mathcal{E}(A)$ is represented in a nice operator basis as
\begin{align}
  \mathcal{E}(A) &=
    \sum_{{i,j=0}}^{J} \boldsymbol{\mathcal{E}}_{ij} \boldsymbol{A}_j F_i\ ,
\end{align}
so that
\bes
\begin{align}
  [\mathcal{E}(A)]^\dagger
  &= \left(\,\,\sum_{{i,j=0}}^{J} 
    \boldsymbol{\mathcal{E}}_{ij} \boldsymbol{A}_j F_i\right)^\dagger
  = \sum_{{i,j=0}}^{J}
    \boldsymbol{\mathcal{E}}_{ij}^* \boldsymbol{A}_j^* F_i^\dagger \\
  &= \sum_{{i,j=0}}^{J}
    \boldsymbol{\mathcal{E}}_{ij} \boldsymbol{A}_j^* F_i 
  = \mathcal{E}(A^\dagger)\ .
\end{align}
\ees

\end{proof}

Thus, after coordinatization, a Hermiticity-preserving superoperator $\mathcal{E}$ can be seen as a real-valued $d^2 \times d^2$-dimensional matrix $\boldsymbol{\mathcal{E}}$.

\subsection{General properties of a Liouvillian}
\label{ss:real-basis}

Since $a$ in \cref{eq:LE} is Hermitian, it can be written as
$a = u^{\dag}\tilde{\gamma}u$ where $u$ is unitary and $\tilde{\gamma}$ is diagonal
with the eigenvalues $\{\gamma_{\a}\}_{\alpha=1}^{J}$ of $a$ on its diagonal.
The $\gamma_{\a}$'s are always real when \cref{eq:LE} is a Markovian quantum master equation. When $a\geq 0$, i.e., when \cref{eq:LE} is a Lindblad equation, they are non-negative. Defining
\beq
L_{\a} = \sum_{j = 1}^{J} u^*_{\a j}F_j\ ,
\label{eq:L_alpha}
\eeq
we have $F_i = \sum_{\alpha = 1}^{J}u_{\a i}L_{\a}$. Then:
\bes
\begin{align}
  \sum_{{i,j=1}}^{J}a_{ij} F_i\rho F_j^\dag
  &= \sum_{{\a,\b=1}}^{J} L_\a \rho L_\b^\dag
    \sum_{{i,j=1}}^{J} u_{\a i} a_{ij} u^\dag_{\b l} \notag \\
&  = \sum_{{\a=1}}^{J} {\g}_\a L_\a\rho L_\a^\dag\ , \\
  \sum_{{i,j=1}}^{J} a_{ij} F_j^\dag F_i
  &= \sum_{{\a,\b=1}}^{J} L_\b^\dag L_\a \sum_{{i,j=1}}^{J} u_{\a i} a_{ij} u^\dgr_{\b l}\notag \\
&  = \sum_{{\a=1}}^{J} {\g}_\a L_\a^\dag L_\a \ .
\end{align}
\ees
Therefore the Markovian quantum master equation [\cref{eq:LE}] becomes:
\bes
\label{eq:Lindblad-eq}
\begin{align}
\dot{\rho} =  {\cal L}\rho &= \mc{L}_H\rho+ \mc{L}_a\rho\\
\mc{L}_H[\cdot] &= -i[H,\cdot]\ ,
\label{eq:L_H2} \\
\mc{L}_a[\cdot] &= \sum_{\a=1}^{J} \g_\a\left( L_{\alpha} \cdot L^{\dag}_{\alpha} -\frac{1}{2}\{L^{\dag}_{\alpha}L_{\alpha},\cdot\}\right)\ .
\label{eq:L_D}
\end{align}
\ees
The $L_{\alpha} \in \mc{B}(\mc{H})$ are called \emph{Lindblad operators}. 

We remark that the representation of the dissipator $\mL_a$ in
the form~\eqref{eq:L_D} need not be unique:
for example, if $a = I$ (the identity matrix), it can be written
as $u^\dag I u$ for any unitary $u$, but [as can be seen from
\cref{eq:L_alpha}] different choices of $u$ lead to different Lindblad operators. More generally, $\tilde{\gamma}$ is unique only up to permutations, and any permutation redefines the Lindblad operators by modifying the diagonalizing unitary matrix $u$. 

\begin{myproposition}
\label{prop:3}
The Liouvillian is Hermiticity-preserving.
\end{myproposition}

\begin{proof}
To show this we can directly compare $[\mathcal{L}(X)]^\dag$ with $\mathcal{L}(X^\dag)$ using \cref{eq:L_H2,eq:L_D}. 
    For the Hamiltonian part, we have the following:
    \bes
    \begin{align}
        [\mL_H(X)]^{\dag} 
        &= i (HX-XH)^{\dag} = i (X^{\dag} H^{\dag} - H^{\dag} X^{\dag}) \\
        &= -i (H X^{\dag} - X^{\dag} H ) = \mL_H(X^{\dag})\ ,
    \end{align}
    \ees
    while for the dissipative part,
    \bes
    \begin{align}
        [\mL_a(X)]^{\dag} 
        &= \sum_{\alpha=1}^{J} \gamma_{\alpha} \left(
        (L_{\alpha}XL_{\alpha}^{\dag})^{\dag} - \frac{1}{2}\{L_{\alpha}^{\dag} L_{\alpha}, X\}^\dag\right) \\
        &= \sum_{\alpha=1}^{J} \gamma_{\alpha} \left(
            (L_{\alpha} X^{\dag} L_{\alpha}^{\dag})
            - \frac{1}{2} \{L_{\alpha}^{\dag} L_{\alpha}, X^{\dag}\}
        \right) \\
        &= \mL_a(X^{\dag})\,.
    \end{align}
    \ees
    Thus, 
\beq
\mathcal{L}\left(X^{\dagger}\right)  = [\mathcal{L}\left(X\right)]^\dag\ .
\label{eq:Herm-pres-2}
\eeq
\end{proof}

\begin{mycorollary}
\label{cor:1}
The Markovian quantum master equation $\dot{\rho} = \mc{L}\rho$ in a nice operator basis
is a real-valued linear ordinary differential equation (ODE) for the vector 
$\boldsymbol{\rho}$ whose coordinates are $\boldsymbol{\rho}_j = \Tr(\rho F_j)$:
\beq
\dot{\boldsymbol{\rho}}=\boldsymbol{\mathcal{L}}\boldsymbol{\rho}\ , \quad \boldsymbol{\rho}\in \mathbb{R}^{d^2}\ , \quad \boldsymbol{\mathcal{L}}\in M(d^2,\mathbb{R})\ .
\label{eq:rhobold}
\eeq
\end{mycorollary}

\begin{proof}
This follows directly from \cref{prop:0,prop:herm-preserving,prop:3} with $X=\rho$ and $\mc{E}=\mc{L}$. 
\end{proof}

\begin{myproposition}
\label{prop:TrL0}
    The Liouvillian maps any operator $X$ to a traceless operator.
    That is, for any operator $X$
    \begin{equation}
        \Tr[\mL(X)] = 0.
    \end{equation}
\end{myproposition}
\begin{proof}
From \cref{eq:LE} we have
\begin{align}
    \Tr[\mL_H(X)] &= -i \Tr([H, X]) \notag{} \\
    &= -i [\Tr(HX) - \Tr(XH)] = 0 ,
\end{align}
and, using $\Tr(AB) = \Tr(BA)$:
\begin{multline}
    \Tr(\mL_a(X))
      = \sum_{i,j=1}^{J} a_{ij} \Tr\left(
        F_i X F_j - \frac12 \left\{F_j F_i, X\right\}\right) \\
    = \sum_{i,j=1}^{J} a_{ij} \left(
        F_j F_i X - \frac12 F_j F_i X - \frac12 F_j F_i X\right)
      = 0.
\end{multline}
\end{proof}

\begin{myproposition}
\label{prop:calL-Hermit}
  Suppose the dissipator $\mL_a$ is 
of the form~\eqref{eq:L_a} with a Hermitian matrix $a$. 
  Then the following conditions are equivalent.
  \begin{enumerate}
    \item $\mL_a$ is Hermitian.
    \item $\mL_a$ can be written in the form~\eqref{eq:L_D} s.t.
    $\forall \alpha\; L_\alpha = L_\alpha^\dagger$.
    \item $a$ is symmetric, i.e., $a=a^T$.
    \item All the matrix elements $a_{ij}$ are real.
  \end{enumerate}  
\end{myproposition}

\begin{proof}
$2\implies 1$. If all the Lindblad operators $L_\alpha$ are Hermitian then
$\mathcal{L}_a^\dag = \mathcal{L}_a$. Indeed,
from \cref{eq:L_D,eq:supop-adjoint2}, for any operator $A$ we have
\begin{align}
  \mL_a^\dag(A)
  &= \sum_{\a=1}^{J} \g_\a\left( L_{\alpha}^{\dag} X L_{\alpha} -\frac{1}{2}\{L^{\dag}_{\alpha}L_{\alpha}, X\}\right) = \mL_a(A).
\end{align}

$3 \Leftrightarrow 4$. Conditions 3 and 4 mean that for each $i,j$ we have $a_{ji} = a_{ij}$ and $a_{ji} = a_{ji}^*$ respectively. Since $a=a^\dag$, the right-hand sides of these two equations are equal for all $i,j$.

$3,4 \implies 2$. We know that $a$ is real-symmetric; hence it
can be diagonalized by an orthogonal matrix. Therefore, $u$ in \cref{eq:L_alpha} can
be chosen to have real matrix elements. For that choice, we have
$\forall \a\;L_\a=L_\a^\dag$. 

$1 \implies 4$. According to 
\cref{eq:supop-adjoint2}, 
$\mathcal{L}_a^\dag = \mathcal{L}_{a^*}$.
From $\mathcal{L}_a = \mathcal{L}_a^\dag$ we have
\begin{equation}
  \mathcal{L}_a = (\mathcal{L}_a + \mathcal{L}_{a^*}) / 2 = \mathcal{L}_{\tilde a},
\end{equation}
where $\tilde{a} = (a + a^*)/2$ --- a matrix with real matrix elements. Later in \cref{thm:forward.inverse}
we will show that $a$ is uniquely determined by $\mL_a$, hence $a = \tilde{a}$ 
and $a = a^*$.

$2 \implies 3$. If $L_\a=L_\a^\dag$ then it can be expanded in the given 
basis $\{F_j\}$ with real coefficients: $L_{\a} = \sum_{j = 1}^{J}w_{\a j}F_j$, 
with $w_{\a j}\in \mathbb{R}$. Substituting this into \cref{eq:L_D} we obtain 
\cref{eq:L_a} with $\tilde{a}=w^T \tilde{\g}w$ in place of $a$ and
$\tilde{\g} = \text{diag}(\g_1,\dots,\g_J)$. Thus $\tilde{a}=\tilde{a}^T$.
Later in \cref{thm:forward.inverse} (see also \cref{prop:Rcond,cor:Rdim})
we will show that $a$ is uniquely determined by $\mL_a$; hence, $a = \tilde{a}$ 
and $a = a^T$.
\end{proof}

By definition, the generator $\mathcal{L}_a$ is unital if $\mathcal{L}_a(I)= 0$~\cite{Lidar200682} [a CPTP map is unital if it maps $I$ to itself: $\mc{E}(I)=I$]. 
Another general property of the Liouvillian is the following:

\begin{myproposition}
\label{prop:a-symm-L-unital}
The Liouvillian is unital if any of the conditions in \cref{prop:calL-Hermit} are satisfied.
\end{myproposition}

\begin{proof}
    It suffices to assume that $a$ is symmetric; then using \cref{eq:L_a}:
    \bes
    \begin{align}
    \label{eq:L_a(I)}
     \mathcal{L}_a(I) &= \sum_{i,j=1}^{J} a_{ij} (F_i F_j-F_j F_i) \\
     &= \sum_{i,j=1}^{J} a_{ij} F_i F_j-\sum_{i,j=1}^{J} a_{ji}F_j F_i = 0 \ .
    \end{align}
    \ees
\end{proof}
Note that the converse is false: unitality does not imply that $a$ must be symmetric (or any of the other conditions in \cref{prop:calL-Hermit}). As a counterexample, consider any nice operator basis containing two or more commuting operators [e.g., the last two matrices in \cref{eq:su(3)-matrices} below]. Then, if $a$ contains a non-zero matrix element only between these two operators [e.g., $a_{78}=1$, all the rest zero, for the nice operator basis in \cref{eq:su(3)-matrices}], \cref{eq:L_a(I)} gives $\mathcal{L}_a(I)=0$, but $a$ is not symmetric.

\subsection{Markovian quantum master equation for the coherence vector: \texorpdfstring{from $H$ and $a$ to $G=Q+R$ and $\vec{c}$}{from H and a to G=Q+R and c}}
\label{ss:LE-for-v}
By \cref{cor:1}, we can expand the density matrix in the nice operator basis as:
\beq
\rho = \vec{F}'\cdot \boldsymbol{\rho}  = \frac{1}{d}I + \vec{F}\cdot\vec{v}\ , 
\label{eq:333}
\eeq
where $\vec{F}' = (F_0,F_1,\dots,F_J)$, $\boldsymbol{\rho} = (1/\sqrt{d},v_1,\dots,v_J)^T$, $\vec{F} = (F_1,\dots,F_J)$ collects the traceless basis-operators into a vector, and the corresponding coordinate vector $\vec{v} = (v_1,\dots,v_J)^T \in \mathbb{R}^{J}$ is called the \emph{coherence vector}. We have:
\beq
v_j = \Tr(\rho F_j) = \boldsymbol{\rho}_j \ .
\label{eq:v_j}
\eeq

We note that the conditions guaranteeing that a given vector represents a valid (i.e.,
non-negative) density matrix are nontrivial. A simple necessary condition
follows from the condition that the purity $\Tr(\rho^2)\le 1$:
\bes
\begin{align}
    \Tr(\rho^2) &= \Tr\left[\left(\frac{1}{d}I + \vec{F}\cdot\vec{v}\right)^2\right] \\
&= \frac{1}{d} + \sum_{i,j=1}^{J} \Tr(F_i F_j)v_i v_j = \frac{1}{d}+\norm{\vec{v}}^2\ ,
\end{align}
\ees
i.e.,
\beq
\label{eq:vnorm}
\|\vec{v}\| \leq \left(1-\frac{1}{d}\right)^{1/2} \ ,
\eeq
which is saturated for pure states. Consequently, as mentioned in \cref{sec:intro}, $G$'s eigenvalues are constrained to have non-positive real parts. This is a necessary condition any candidate \cref{eq:3} must satisfy to represent a Lindblad equation. Additional inequalities have been derived in Ref.~\cite{byrd:062322}.

\begin{myproposition}
\label{prop:4}
The coherence vector satisfies \cref{eq:3}. 
Moreover, the decomposition of $\mc{L}$ as $\mc{L} = \mc{L}_H+\mc{L}_a$ with $\mc{L}_H$ and $\mc{L}_a$ given by Eqs.~\eqref{eq:L_H} and~\eqref{eq:L_a}
induces the decomposition of $G$ in \cref{eq:3} into $G=Q+R$, where $\mc{L}_H[\rho] \leadsto Q\vec{v}$, $\mc{L}_a[\rho] \leadsto R\vec{v} + \vec{c}$, $\vec{c}\in\mathbb{R}^J$, $R,Q\in M(J,\mathbb{R})$, and $Q=-Q^T$ (antisymmetric). 
\end{myproposition}

\begin{proof}
After separating out $v_0 = \boldsymbol{\rho}_0=1/\sqrt{d}$, Eq.~\eqref{eq:rhobold} becomes
\beq
\dot{v}_i = \sum_{j=1}^{J} \boldsymbol{\mathcal{L}}_{ij} v_j + \frac{1}{\sqrt{d}}\boldsymbol{\mathcal{L}}_{i0} \ ,\quad 1\le  i\le J \ .
\label{eq:dotv_j}
\eeq
This already establishes that $\vec{v}$ satisfies \cref{eq:2}, with
\bes
\begin{align}
Q_{ij} + R_{ij} &\equiv  \boldsymbol{\mathcal{L}}_{ij} =  \Tr[F_{i}\mathcal{L}(F_{j})] \equiv G_{ij}\\
\label{eq:cjdef}
c_j &\equiv \frac{1}{\sqrt{d}}\boldsymbol{\mathcal{L}}_{j0} = \frac{1}{\sqrt{d}}(Q_{j0}+R_{j0}) \ ,
\end{align}
\ees
for $1\le  i,j\le J$.

To prove the remaining claims, recall that $\boldsymbol{\mathcal{L}}_{ij}= \boldsymbol{\mathcal{L}}_{ij}^*$ by \cref{cor:1}. Therefore, by linearity, for $1\le  i,j\le J$: 
\bes
\label{eq:8.4.2}
\begin{align}
\label{eq:Qjk}
Q_{ij} &\equiv \Tr[F_{i}\mathcal{L}_H(F_{j})] = Q_{ij}^*\\
\label{eq:Rjk}
R_{ij} &\equiv \Tr[F_{i}\mathcal{L}_a(F_{j})]= R_{ij}^*\ ,
\end{align}
\ees
i.e., $R,Q\in M(J,\mathbb{R})$. In addition,
\bes
\begin{align}
Q_{ji} &= \Tr[F_{j}\mathcal{L}_H(F_{i})] = -i \Tr\left(F_{j}[H,F_{i}]\right) 
\\
&
= i \Tr\left(F_{i}[H,F_{j}]\right) = -\Tr[F_{i}\mathcal{L}_H(F_{j})] = -Q_{ij}\ ,
\end{align}
\ees
i.e., $Q$ is antisymmetric. Moreover, by expanding $H$ in the nice operator basis as 
\beq
\label{eq:Hfromh}
H = \sum_{m=0}^{J} h_m F_m\ ,
\eeq
we can write $Q$'s matrix elements as:
\bes
\label{eq:QfromH}
\begin{align}
\label{eq:QfromH-a}
Q_{ij} &= -i \Tr\left(F_{i}[H,F_{j}]\right) \\
\label{eq:QfromH-b}
&= -i \sum_{m=1}^J h_m \Tr \left( F_i[F_m,F_j]\right) \ .
\end{align}
\ees

As for $\vec{c}$, since $F_0 = I/\sqrt{d}$, we have
\beq
Q_{j0} = \frac{-i}{\sqrt{d}}\Tr\left(F_{j}[H,I]\right) =0,
\eeq
so that for $1\le  j\le J$, \cref{eq:cjdef} is replaced by:
\begin{align}
\label{eq:c_j}
c_j = \frac{1}{d}\Tr[F_{j}\mathcal{L}_a(I)] = \frac{1}{\sqrt{d}}R_{j0} = 
c_j^*\ ,
\end{align}
i.e., $\vec{c}\in\mathbb{R}^J$.
\end{proof}

Let $\mathrm{spec}(A)$ denote the spectrum of the operator $A$, i.e., the set of its eigenvalues. 
\begin{myproposition}
\label{prop:specL=G}
\beq
\mathrm{spec}(\mL) = \{0\}\cup\mathrm{spec}(G) .
\eeq
\end{myproposition}

\begin{proof}
Note that, as follows from proposition \ref{prop:TrL0},
\begin{equation}
  \boldsymbol{\mathcal{L}}_{0j} = 0
\quad \text{i.e.},\quad \boldsymbol{\mathcal{L}}_0 = \vec{0}\ .
\end{equation}

Combining this with \cref{eq:8.4.2}, we see that in the basis $\{F_j\}_{j=0}^{J}$:
\beq
  \label{eq:Lij-matrix-form}
  \boldsymbol{\mathcal{L}} = 
  \left(
  \begin{array}{cc}
  0  & 0 \cdots 0   \\
  \sqrt{d}\vec{c}  & G    
  \end{array}
  \right)\,.
\eeq
Computing the spectrum of the superoperator $\mL$ is equivalent to finding the eigenvalues of $\boldsymbol{\mathcal{L}}$, which are the solutions of the characteristic equation
\beq
\det\left(\lambda I_{J+1}-\boldsymbol{\mathcal{L}}\right) = \lambda \det\left(\lambda I_{J}-G\right) = 0 ,
\eeq
where $I_n$ denotes the $n\times n$ identity matrix. These solutions are $\lambda = 0$ and $\mathrm{spec}(G)$.
\end{proof}

\subsection{When does \texorpdfstring{$\vec{c}$}{c} vanish?}

\begin{myproposition}
\label{prop:5}
$\vec{c}=\vec{0}$ iff $\mathcal{L}_a$ is unital.
\end{myproposition}

\begin{proof}
The proof is immediate from \cref{eq:Lij-matrix-form}, which shows that $\vec{c}$ is the vector representing the matrix $\mathcal{L}_a(I)/d$. Thus $\vec{c}=\vec{0}$ iff $\mathcal{L}_a(I)=0$.

A more explicit proof is the following calculation. First, if $\mathcal{L}_a$ is unital then, using Eqs.~\eqref{eq:Rjk} and~\eqref{eq:c_j}, we have:
\beq
\label{eq:c_j-unital}
c_k = \frac{1}{d}\Tr[F_{k}\mathcal{L}_a(I)] = 0\quad (\mL_a \text{ unital})\ .
\eeq
On the other hand, if $\vec c=0$, 
\begin{align}
\label{eq:unital-c_j}
\Tr[F_{k}\mathcal{L}_a(I)] = d c_k = 0,
\end{align}
hence $\mathcal{L}_a(I)=0$.
\end{proof}
Note that using \cref{eq:L_a} we can write down the general formula for $\vec{c}$'s elements in a given nice operator basis:
\begin{equation}
\label{eq:cfroma}
c_k =  \frac{1}{d} \sum_{i,j=1}^{J} a_{ij}\Tr(\big[F_i , F_j \big]F_k)\,.
\end{equation}

\subsection{When is \texorpdfstring{$R$}{R} symmetric?}
Here by $R$ we mean the $J\times J$
matrix with elements $\{R_{kl}\}_{k,l=1}^J$. Using \cref{eq:Rjk}:
\bes
\label{eq:R}
\begin{align}
  \label{eq:LatoR}
  R_{kl} &= \Tr\big[F_k \mc{L}_a(F_l)\big]  \\
  \label{eq:Rfroma}
  &= \sum_{{i,j=1}}^{J} a_{ij}\Tr\big[F_k \big(F_i F_l F_j   -  \frac{1}{2} \left\{F_j F_i, F_l\right\} \big)\big] \ .
\end{align}
\ees
Note that this is the general formula for $R$'s elements in a given nice operator basis.
$R$ is not symmetric in general. However, we have two special cases presented in the following Propositions.

\begin{myproposition}
\label{prop:R-su2}
$R$ is symmetric in the single qubit ($d=2$) case.
\end{myproposition}
\begin{proof}
By direct calculation, it follows that if we choose the traceless elements of the nice operator basis as the normalized Pauli matrices $\{\sigma^x,\sigma^y,\sigma^z\}/\sqrt{2}$, then $R_{kl}=R_{lk}$ for arbitrary $a$.

An alternative way to see this is from \cref{cor:Rdim} below: the dimension of the space of possible antisymmetric components of $R$ is $J(J-3)/2$, which is $0$ iff $J\in\{0,3\}$, i.e., when $d\in\{1,2\}$.
\end{proof}

\begin{myproposition}
\label{prop:Rsymm-LHerm}
The following statements are equivalent:
\begin{enumerate}
  \item $R$ is symmetric and $\vec{c}=\vec{0}$;
  \item Any of the conditions in \cref{prop:calL-Hermit} are satisfied.
\end{enumerate}
\end{myproposition}

\begin{proof}
If $\mL_a^{\dag} = \mL_a$ (one of the four equivalent conditions in \cref{prop:calL-Hermit}) then, using \cref{eq:Rjk}:
\bes
\begin{align}
    R_{lk} &= \langle F_l, \mL_a(F_k) \rangle 
    =  \langle \mL_a(F_l), F_k \rangle \\
    & =   \langle F_k, \mL_a(F_l)  \rangle^* = R_{kl}^*
    = R_{kl}\ ,
\end{align}
\ees
where we also used that the matrix elements of $R$ are real (\cref{prop:4}). This proves that 2 $\implies$ 1.
Similarly, 
\bes
\begin{align}
    c_k &= \frac{1}{d} \left<F_k, \mathcal{L}_a(I)\right>
    = \frac{1}{d} \left<\mathcal{L}_a(F_k),I\right> \\
    &= \frac{1}{d}\Tr[\mathcal{L}_a(F_k)]^* = 0 ,
\end{align}
\ees
where in the last equality we used \cref{prop:TrL0}.

The proof that 1 $\implies$ 2 is almost identical:
\bes
\begin{align}
    \langle F_l, \mL_a(F_k) \rangle = R_{lk} =  R_{kl} = R_{kl}^* &= \langle F_k, \mL_a(F_l) \rangle^* \\
    &=  \langle \mL_a(F_l) , F_k  \rangle \ ,
\end{align}
\ees
and
\begin{equation}
    \left<F_k, \mathcal{L}_a(F_0)\right> = \sqrt{d} c_k
    = 0 = \left<\mathcal{L}_a(F_k),F_0\right>,
\end{equation}
which means $\mL_a^{\dag} = \mL_a$.
\end{proof}

In particular, $a=a^T$ implies that $R$ is symmetric. Note that, consistent with the comment following \cref{prop:a-symm-L-unital}, the converse is false, i.e., $R$ being symmetric does not imply that $a$ is real-symmetric: any $a$ yielding a symmetric $R$ and non-zero $\vec{c}$ is a counterexample. As we describe
in \cref{thm:forward.inverse}, from any such pair $(R,c)$
one could recover $a$ via \cref{eq:Gctoa}, which would be a counterexample. For a specific counterexample, see \cref{sss:amplitude}.

According to \cref{cor:Rdim} below, the dimension of the space of possible antisymmetric
components of $R$ is $J(J-3)/2$, which is non-zero if and only if $d \geq 3$. Thus, for
$d\geq 3$ the matrix $R$ is not always symmetric. An explicit example of this for $d=3$
is the following. Consider a qutrit subject to amplitude damping (spontaneous emission)
involving just two
of the three levels: 
\begin{align}
\label{eq:a-su3}
    a = \left(
\begin{array}{ccc}
 1 & -i &   \\
 i & 1 &  \\
 & & 0_{6\times 6}
\end{array}
\right)\ ,
\end{align}
which is a non-negative matrix with eigenvalues $0$ ($7$-fold degenerate) and $2$. Note that in all the examples given in this work, we use a specific choice of a nice operator basis $\{F_i\}$: the generalized Gell-Mann matrices normalized to satisfy the normalization condition Eq.~\eqref{eq:F-reqs}; see \cref{ss:nice-op} and Ref.~\cite{ODE-to-GKLS-supporting-calcs}.
By \cref{eq:R} we have
\begin{align}
\label{eq:R-su3}
    R = \left(
\begin{array}{cccccccc}
 -\frac{1}{4} & 0 & 0 & 0 & -\frac{1}{4} & 0 & 0 &
   0 \\
 0 & -\frac{1}{4} & 0 & \frac{1}{4} & 0 & 0 & 0 &
   0 \\
 0 & 0 & -\frac{1}{2} & 0 & 0 & 0 & 0 & 0 \\
 0 & -\frac{3}{4} & 0 & -\frac{5}{4} & 0 & 0 & 0 &
   0 \\
 \frac{3}{4} & 0 & 0 & 0 & -\frac{5}{4} & 0 & 0 &
   0 \\
 0 & 0 & 0 & 0 & 0 & -\frac{1}{2} & \frac{1}{4} &
   \frac{1}{4 \sqrt{3}} \\
 0 & 0 & 0 & 0 & 0 & -\frac{3}{4} & -\frac{5}{4} &
   -\frac{\sqrt{3}}{4} \\
 0 & 0 & 0 & 0 & 0 & -\frac{\sqrt{3}}{4} &
   -\frac{\sqrt{3}}{4} & -\frac{3}{4} \\
\end{array}
\right)
\end{align}
which is not symmetric. We revisit this example in \cref{sss:qutrit-amplitude-damping}.

\subsection{Properties of the linear map
\texorpdfstring{$a \mapsto (R,\vec{c})$}{a to (R,c)}}
\label{ss:atoRc}
Consider a map $\mathcal{F}: a \mapsto (R, \vec{c})$ defined by \cref{eq:R,eq:cfroma}. This is a linear map of real vector spaces $\mathcal{F}: M_{\textrm{sa}}(J,\mathbb{C}) \to M(J,\mathbb{R})\oplus \mathbb{R}^J$, where $M_{\textrm{sa}}(J,\mathbb{C})$ is the $\mathbb{R}$-vector space of
Hermitian (or self-adjoint; hence the ``sa'' subscript) $J\times J$ matrices over $\mathbb{C}$.

Then, as will follow later from
\cref{thm:forward.inverse,prop:GtoHconsistency},
this map is injective, i.e., $a\neq a' \implies \mathcal{F}(a) \neq \mathcal{F}(a')$. In other words, if the pairs $(R, \vec{c})$ and $(R', \vec{c}')$ are equal, then also the corresponding rate matrices $a$ and $a'$ are equal.
A direct way to
prove the injectivity of $\mathcal{F}$, which does not
rely on \cref{thm:forward.inverse} (but in part uses similar ideas), 
is presented in \cref{sec:lemmas12}.

The following Lemma, together with \cref{cor:Rdim} below,
describes the image of $\mathcal{F}$.
\begin{mylemma}
\label{lemma4}
For all symmetric $R$, $R_{\text{sym}}$: $(R_{\text{sym}},\vec{0}) \in \text{Im}(\mc{F})$. In particular, $(R_{\text{sym}},\vec{0})$ is the image of some real-valued symmetric $a$.
\end{mylemma}
\begin{proof}
  The subspace of real-valued symmetric $a$ in $M(J,\mathbb{R})$ is of dimension $J(J+1)/2$. We know by \cref{prop:Rsymm-LHerm} that $\mc{F}$ maps real-symmetric $a$ to $\vec{c}=\vec{0}$ and $R$ real-symmetric. The subspace $V = \{(R_{\text{sym}},\vec{0})\}$ has dimension $J(J+1)/2$ in the codomain of $\mc{F}$. Since, as explained above, $\mc{F}$ is injective, the image of all real-symmetric $a$ is the subspace $V$.
\end{proof}

\subsection{Properties of a nice operator basis}
\label{ss:nice-op}
In this subsection, we define structure constants corresponding to a nice operator basis and
provide alternative forms of \cref{eq:cfroma,eq:QfromH-b} using these structure constants.
We note that the structure constants are often explicitly defined for the case of
(normalized and generalized) Gell-Mann matrices~\cite{generalized-Gell-Mann}. As mentioned above, this is our choice in all the examples in this work. However, the theory presented here applies to any choice of a nice
operator basis $\{F_i\}_{i=0}^{J}$ (see \cref{def:nice-operator-basis}).

For any such choice, the elements $\{F_i\}_{i=1}^J$ form a generator set of the Lie algebra $\mathfrak{su}(d)$.\footnote{Not every generator set of a Lie algebra is a nice operator basis. For example, the set $\{\sigma^x,\sigma^y,\sigma^z\}$ is a generator set of $\mathfrak{su}(2)$ but it 
violates the normalization conditions \cref{eq:F-reqs}; indeed, with
this choice we have
$\left<F_i,F_j\right> = 2 \delta_{ij}$.} We can define structure constants $f_{ijk}$ via
\beq
  [F_i,F_j] = i \sum_{k=1}^{J} f_{ijk} F_k = i f_{ijk} F_k\ ,
  \label{eq:F-Lie}
\eeq
where in the second equality we used the Einstein summation convention of summing over repeated indices, which we use henceforth when convenient.

The structure constants are totally antisymmetric, i.e., $f_{jkl}=-f_{kjl}=f_{klj}$ and themselves satisfy a type of orthogonality relation~\cite{Kaplan:1967aa,Haber:16,sun-norm-comment}:
\beq
\label{eq:ff}
f_{jkl} f_{jkm} = 2d \d_{lm}\ .
\eeq

Using \cref{eq:F-reqs,eq:F-Lie} we can then further simplify \cref{eq:cfroma}:
\begin{align}
\label{eq:c2}
c_k =  \frac{1}{d} i a_{ij}f_{ijk}  \ .
\end{align}

We may also derive an explicit expression relating the coordinates $h_m$ of the Hamiltonian in the expansion~\eqref{eq:Hfromh} to the matrix elements of $Q$. Namely, inserting \cref{eq:F-Lie} into \cref{eq:QfromH}, we have:
\begin{align}
Q_{jk} &= -i h_l \Tr(F_j[F_l,F_k]) = -i h_l if_{lkj} = -f_{jkl} h_l\ .
\end{align}
On the other hand, using
\beq
f_{jkm}Q_{jk} = -f_{jkl}f_{jkm}h_l = -2d \d_{lm} h_l = -2dh_m \ ,
\eeq
we have
\beq
\label{eq:hm}
h_m = -\frac{1}{2d} f_{jkm}Q_{jk} \ .
\eeq
This implies that $H$ can be computed given $Q$ as
\beq
\label{eq:HfromQ}
H = -\frac{1}{2d} f_{jkm}Q_{jk}F_m \ .
\eeq
We discuss the consistency between this expression and \cref{eq:GctoH} in
\cref{prop:recoverHconsistency1}.

\section{Solution of the inverse problem}
\label{sec:sol}
We now set up the necessary mathematical framework to solve the
inverse problem in full generality. Since some of 
the results after \cref{prop:TrL0} used 
\cref{thm:forward.inverse} and its corollaries, in this section we
only use the results up to (and including) \cref{prop:TrL0}, to avoid any circular references.

\subsection{Forward and inverse transformations}
For a field $F \in \{\mathbb{R}, \mathbb{C}\}$ let $M_{0}(d, F)$
denote the subspace of traceless matrices. Let $M_{\textrm{sa}}(d, F)$ denote
the $\mathbb{R}$-subspace of Hermitian matrices.
E.g., $M_{\textrm{sa}}(d, \mathbb{R})$ is a vector space
of real-valued $d\times d$ symmetric matrices.
Finally, let $M_{0,\textrm{sa}}(d, F)$ be the $\mathbb{R}$-subspace
of matrices which are both traceless and Hermitian.
\begin{mytheorem}
  \label{thm:forward.inverse}
  For any nice operator basis $\{F_n\}_{n=1}^{J}$ 
  there is a linear bijective correspondence between
  the following objects.
  \begin{enumerate}
    \item Pairs $(H, a)$ where $H \in M_{0,\textrm{sa}}(d, \mathbb{C})$,
      $a \in M_{\textrm{sa}}(J, \mathbb{C})$.
    \item $d \times d \times d \times d$ tensors
      $x = \{x_{ijkl}\}_{i,j,k,l=1..d}$ over
      $\mathbb{C}$ satisfying
      \begin{equation}
        \label{eq:xcond}
        x_{ijkl} = x_{lkji}^*,\quad \sum_{k=1}^d (x_{ijkk} + x_{kkij}) = 0\ .
      \end{equation}
    \item $d \times d \times d \times d$ tensors
      $\tilde{x} = \{\tilde{x}_{ijkl}\}_{i,j,k,l=1..d}$ over
      $\mathbb{C}$ satisfying
      \begin{equation}
        \label{eq:xt.cond}
        \tilde{x}_{ijkl} = \tilde{x}_{lkji}^*\ ,
        \quad \sum_{i=1}^{d} \tilde{x}_{ijki} = 0\ .
      \end{equation}
    \item Linear Hermiticity-preserving superoperators $\mathcal{L}$ acting on $M(d,\mathbb{C})$
      satisfying
      \begin{equation}
        \label{eq:Lcond}
        \mathcal{L}\left(X^\dagger\right) = [\mathcal{L}(X)]^\dagger\ ,\quad
        \Tr[\mathcal{L}(X)] = 0
      \end{equation}
      for each $X \in M(d,\mathbb{C})$.
    \item Linear superoperators
      \begin{equation}
        \mathcal{L}: M_{sa}(d, \mathbb{C}) \to M_{0,sa}(d, \mathbb{C})\ .
      \end{equation}
    \item Pairs $(G, \vec{c})$ where $G \in M(J, \mathbb{R})$,
      $c \in \mathbb{R}^{J}$.
  \end{enumerate}
  The corresponding transformations are given as follows:
  \begin{itemize}
    \item $1\to 2$:
      \begin{multline}
        \label{eq:1to2}
        x_{ijkl} = -iH_{ij} \delta_{kl} + i \delta_{ij} H_{kl} 
        + \sum_{m,n=1}^{J} a_{mn} (F_m)_{ij} (F_n)_{kl}\ .
      \end{multline}
    \item $2\to 1$:
    \bes
      \begin{align}
        \label{eq:2to1.H}
        H_{ij} &= \frac1{2id}\sum_{k=1}^{d} (x_{kkij} - x_{ijkk})\ ,\\
        \label{eq:2to1.a}
        a_{mn} &= \sum_{i,j,k,l=1}^{d} (F_{m})_{ji} x_{ijkl} (F_{n})_{lk}\ .
      \end{align}
      \ees
    \item $2 \to 4,5$:
      \begin{equation}
        \label{eq:2to4}
        [\mathcal{L}(X)]_{il} = \sum_{j,k=1}^{d}
        (x_{ijkl} X_{jk} - \frac12 x_{jkij} X_{kl} -
        \frac12 X_{ij} x_{kljk})\ .
      \end{equation}
    \item $4 \to 5$ is simply the restriction to Hermitian operators
      $X$.
    \item $5 \to 4$:
      \begin{equation}
        \mathcal{L}(X) = \mathcal{L}[\real(X)]
        + i\mathcal{L}[\imag(X)]\ ,
      \end{equation}
      where
      \begin{equation}
        \label{eq:real,imag}
        \real(X) = \frac{X + X^\dagger}{2}\ , \qquad
        \imag(X) = \real(X/i) = \frac{X - X^\dagger}{2i}\ .
      \end{equation}
    \item $1 \to 4$ is given by Eqs.~\eqref{eq:LEa}-\eqref{eq:L_a}.
    \item $5 \to 6$:
    \bes
       \label{eq:5to6}
       \begin{align}
        G_{nm} &= \Tr[F_n\mathcal{L}(F_m)]\ , \\
        c_n &= \frac{1}{\sqrt{d}}\Tr[F_n\mathcal{L}(F_0)]\ .
      \end{align}
      \ees
    \item $6 \to 4,5$:
      \begin{equation}
        \label{eq:6to4}
        \mathcal{L}(X) = \sum_{n=1}^{J} \left(
          \sum_{m=1}^{J} G_{nm} \Tr(F_m X) + c_{n} \Tr(X)
        \right) F_{n}\ .
      \end{equation}
    \item $4 \to 3$:
      \begin{equation}
        \label{eq:4to3}
        \tilde{x}_{ijkl} = \left[\mathcal{L}(\ket{j}\!\!\bra{k})\right]_{il}\ .
      \end{equation}
    \item $3 \to 1$ is given by Eq.~\eqref{eq:2to1.H} and Eq.~\eqref{eq:2to1.a}
      with $\tilde x$ in place of $x$.
    \item $3 \to 4$ is given by Eq.~\eqref{eq:2to4}
      with $\tilde x$ in place of $x$.
    \item $3 \to 2$ is given by
      \begin{equation}
        \label{eq:3to2}
        x_{ijkl} = \tilde{x}_{ijkl}
          - b_{ij} \delta_{kl} - \delta_{ij} b_{kl}\ ,
      \end{equation}
      where
      \begin{equation}
        \label{eq:3to2.b}
        b_{ij} = \frac{1}{2d}\sum_{k=1}^d \left(\tilde{x}_{ijkk} + \tilde{x}_{kkij} - \frac{1}{d}\sum_{l=1}^d\tilde{x}_{llkk} \delta_{ij}\right)\ .
      \end{equation}
  \end{itemize}
\end{mytheorem}
We denote the spaces in \cref{thm:forward.inverse} as $\mathcal{V}_1,\dots,\mathcal{V}_6$
and the corresponding maps as $\varphi_{ij}: \mathcal{V}_j \to \mathcal{V}_i$. \cref{thm:forward.inverse} and \cref{lm:GctoHa} below imply that the diagram in \cref{fig:V1V6} is commutative and all its arrows are
$\mathbb{R}$-linear bijections (given by the formulas in
\cref{thm:forward.inverse} and \cref{lm:GctoHa}). Once
the theorem is proven, we will use the maps $\varphi_{ij}$
for any indices $i,j=1,\dots,6$. Such maps are defined as
the compositions of the maps in \cref{fig:V1V6};
any such composition will
provide the same result due to the commutativity of \cref{fig:V1V6}.
\begin{figure}
\begin{tikzpicture}
  \newlength{\dx}
  \newlength{\dy}
  \setlength{\dx}{2cm}
  \setlength{\dy}{2cm}
  \node[draw,circle,align=center](n1) at (-2\dx, 0)
    {$\mathcal{V}_1$\\$(H,a)$};
  \node[draw,circle,align=center](n2) at (-1\dx, 0)
    {$\mathcal{V}_2$\\$x$};
  \node[draw,circle,align=center](n3) at (-\dx,\dy)
    {$\mathcal{V}_3$\\$\tilde{x}$};
  \node[draw,circle,align=center](n4) at (    0, 0)
    {$\mathcal{V}_4$\\$\mathcal{L}$};
  \node[draw,circle,align=center](n5) at (   0,\dy)
    {$\mathcal{V}_5$\\$\mathcal{L}$};
  \node[draw,circle,align=center](n6) at ( \dx,\dy)
    {$\mathcal{V}_6$\\$(G,\vec{c})$};
  \draw[stealth-stealth] (n1) -- (n2);
  \draw[stealth-stealth] (n3) -- (n4);
  \draw[stealth-stealth] (n4) -- (n5);
  \draw[-stealth] (n3) -- (n1);
  \draw[-stealth] (n3) -- (n2);
  \draw[-stealth] (n2) -- (n4);
  \draw[-stealth] (n5) -- (n6);
  \draw[-stealth] (n6) -- (n4);
  \path[-stealth] (n1) edge[bend right] (n4);
  \path[dashed,-stealth] (n6) edge[bend left] (n5);
  \draw[dashed,-stealth] (n2) -- (n5);
  \path[dotted,-stealth] (n6) edge[bend right] (n3);
  \draw[dotted,-stealth] (n6.south)
    ..controls (0.5\dx,-0.7\dy) and (-\dx,-0.7\dy).. (n1.south east);
\end{tikzpicture}
\caption{A commutative diagram representing the transformations between the various spaces described in \cref{thm:forward.inverse} (solid lines) and Lemma~\ref{lm:GctoHa} (dotted lines). These transformations capture the equivalence between the algebraic objects appearing in Markovian quantum master equations and first-order differential equations.
The maps to $\mathcal{V}_5$ from $\mathcal{V}_2$ and $\mathcal{V}_6$ are
represented by dashed lines as they are represented by the same
formulas \eqref{eq:2to4} and \eqref{eq:6to4} as the corresponding
maps to $\mathcal{V}_4$: indeed, these maps to $\mathcal{V}_5$ trivially
coincide with the composition of the corresponding maps to
$\mathcal{V}_4$ with the restriction $\mathcal{V}_4 \to \mathcal{V}_5$.}
\label{fig:V1V6}
\end{figure}
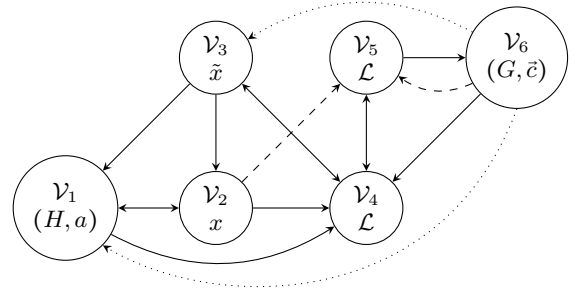
  
\begin{proof}
  First, note that the dimension of each of the $\mathbb{R}$-vector spaces
  $\mathcal{V}_i$ is $(d^2-1)d^2 = J(J+1)$.
  
  We would like to prove that all maps are well
  defined and that the diagram is commutative,
  i.e., if we start at any object in any node and apply
  the maps in the diagram until we end up in the same node, we will have 
  obtained the same object as the one we started with.

  For $\varphi_{21}\colon \mathcal{V}_{1} \to \mathcal{V}_2$
  we need to show that the right-hand side of Eq.~\eqref{eq:1to2}
  satisfies Eq.~\eqref{eq:xcond}.
  We check this by direct computation:
  \begin{align}
    x_{lkji}^* &= iH_{lk}^* \delta_{ji} - i \delta_{lk} H_{ji}^*
      + \sum_{m,n=1}^{J} a_{mn}^* (F_m)_{lk}^* (F_n)_{ji}^* \notag \\
    &= - i H_{ij}\delta_{kl}  + i\delta_{ij} H_{kl}
      + \sum_{m,n=1}^{J} a_{mn} (F_m)_{ij} (F_n)_{kl}\notag \\
    &= x_{ijkl}\ .
  \end{align}
  Here we used that $H, a, \{F_n\}$ are Hermitian, swapped the order of
  two terms with $H$, and renamed the indices $n\leftrightarrow m$
  in the summation.
  To check the second property of $x$ we note that the $\{F_n\}$ are
  traceless. Hence, the second term in Eq.~\eqref{eq:1to2} does not contribute
  and
  \begin{align}
&    \sum_{k=1}^d (x_{ijkk} + x_{kkij}) = \\
&    \quad \sum_{k=1}^d \Bigl(
        -iH_{ij} \delta_{kk} + i \delta_{ij} H_{kk} 
        -iH_{kk} \delta_{ij} + i \delta_{kk} H_{ij}\Bigr)
    = 0\ . \notag
  \end{align}

  For $\varphi_{12}\colon \mathcal{V}_2 \to \mathcal{V}_1$
  the fact that $H$ and $a$ are Hermitian follows
  from the first property of $x$ in \cref{eq:xcond} and 
  \cref{eq:2to1.H,eq:2to1.a} for $H$ and $a$.
  The trace of the right-hand side of Eq.~\eqref{eq:2to1.H} evaluates to 0.

  If we start with $(H,a)$ and apply the maps
  $\mathcal{V}_1\to \mathcal{V}_2 \to \mathcal{V}_1$
  given by \cref{eq:1to2,eq:2to1.H,eq:2to1.a} we
  obtain some $(H',a') = \varphi_{12}(\varphi_{21}((H,a)))$.
  Let us prove that these match the original $(H,a)$, i.e., that
  $\varphi_{12}\circ \varphi_{21} = \id_{\mathcal{V}_1}$
  Since $F_m$ and $F_n$ are traceless, the term corresponding to $a$ in the
  expression for $H'$ is zero. Using the fact that $\Tr(H)=0$ and $\Tr(I)=d$
  we obtain
  \bes
  \begin{align}
    H'_{ij} &= \frac1{2id}\sum_{k=1}^{d} (x_{kkij} - x_{ijkk}) \\
    & = \frac1{2id}\Bigl(
      -i\Tr(H) \delta_{ij} + i \Tr(I) H_{ij}
      +iH_{ij} \Tr(I) \notag \\
      &\qquad - i \delta_{ij} \Tr(H)
    \Bigr) \\
    & = \frac1{2id}\left(id H_{ij} + iH_{ij}d\right) = H_{ij}\ .
  \end{align}
  \ees
  Similarly, we can observe that the term with $H$ from Eq.~\eqref{eq:1to2}
  does not contribute to $a'$ since the $\{F_n\}$ are traceless. Using the fact
  that the $\{F_{m}\}$ are Hermitian and form an orthonormal basis we obtain
  \begin{equation}
    a'_{mn} = \sum_{i,j=1}^{J} a_{ij} \Tr(F_{m} F_{i}) \Tr(F_{n} F_{j})
    = a_{mn}\ .
  \end{equation}

  It follows that
  $\varphi_{21} \circ \varphi_{12} = \id_{\mathcal{V}_2}$.
  Indeed, it is a general fact that a left inverse of a linear map
  between 2 vector spaces of the same dimension is necessarily its 
  right inverse as well~\cite[Proposition 1.20]{oliveira2022linear}.
  We also use this argument for other loops, 
  proving the commutativity of each loop in the diagram
  for only a single starting point.

  By direct comparison of Eqs.~\eqref{eq:LEa}-\eqref{eq:L_a} with
  Eq.~\eqref{eq:1to2} and Eq.~\eqref{eq:2to4} one can conclude that 
  they yield the same operator $\mathcal{L}$. To check that 
  $\varphi_{41}$ is well defined, one would need to check that
  the operator $\mL$ defined by Eqs.~\eqref{eq:LEa}-\eqref{eq:L_a}
  satisfies Eq.~\eqref{eq:Lcond}. This was already done in
 \cref{prop:3} and \cref{prop:TrL0}. Therefore,
  $\varphi_{42} \circ \varphi_{21} = \varphi_{41}$. Note, that
  $\varphi_{42}$ is well defined because for any
  $x \in \mathcal{V}_2$ we have
  \begin{equation}
    \varphi_{42}(x) = (\varphi_{42}\circ \varphi_{21} \circ \varphi_{12})(x) = \varphi_{41}(\varphi_{12}(x)).
  \end{equation}
  
  The maps
  $\varphi_{45},\,\varphi_{54}: \mathcal{V}_5 \leftrightarrow \mathcal{V}_4$
  are clearly well-defined. Also
  $\varphi_{54} \circ \varphi_{45} = \id_{\mathcal{V}_5}$
  and, hence, $\varphi_{45} \circ \varphi_{54} = \id_{\mathcal{V}_4}$. 

  The fact that the elements of $G$ and $\vec{c}$ are real follows from
  Eq.~\eqref{eq:5to6}, the fact that the $\{F_n\}$ form an orthonormal basis,
  and the fact that $\mathcal{L}$ maps Hermitian operators to Hermitian
  operators. Let us check that the composition
  $\varphi_{46} \circ \varphi_{65} \circ \varphi_{54} = \id_{\mathcal{V}_4}$.
  Denote the resulting operator given by Eq.~\eqref{eq:6to4} as
  $\mathcal{L}' = (\varphi_{46} \circ \varphi_{65} \circ \varphi_{54})
  (\mathcal{L})$. Since $\{F_n\}_{n=0,\dots,J}$ form a basis in
  $M(d,\mathbb{C})$ and $\mathcal{L}$ and both $\mathcal{L}'$ are
  $\mathbb{C}$-linear, it is sufficient to check that
  $\mathcal{L}(F_l) = \mathcal{L}'(F_l)$ for $l=0,\dots,J$.
  For $l=0$ we have, using Eq.~\eqref{eq:6to4}:
  \bes
  \begin{align}
    \mathcal{L}'(F_0) &= \sum_{n=1}^{J} c_{n}\Tr(I/\sqrt{d})F_{n} = \sum_{n=1}^{J} c_{n}\sqrt{d}F_{n} \\
    &= \sum_{n=1}^{J} \Tr[F_n\mathcal{L}(F_0)] F_{n} = \mathcal{L}(F_0).
  \end{align}
  \ees
  And for $n>0$ we have:
  \bes
  \begin{align}
    \mathcal{L}'(F_l) &= \sum_{n,m=1}^J G_{nm}\Tr(F_{m}F_l)F_n = \\
    & = \sum_{n=1}^{J} G_{nl}F_{n} 
    = \sum_{n=1}^{J} \Tr[F_n\mathcal{L}(F_l)] F_{n} \\
    & = \mathcal{L}(F_l)\ .
  \end{align}
  \ees
  Here we again used that $\{F_n\}_{n=0,\dots,J}$ are Hermitian
  and form an orthonormal basis,
  and that $\Tr[F_0\mathcal{L}(X)]=0$ for any $X$ [\cref{eq:Lcond}]. 

  Now consider the map $\varphi_{34}: \mathcal{L} \mapsto \tilde{x}$.
 \cref{eq:xt.cond} follows from \cref{eq:Lcond} applied to
  $X=\ket{j}\!\!\bra{k}$. Let us prove that
  $\varphi_{23}\circ\varphi_{34}\circ\varphi_{42} = \id_{\mathcal{V}_2}$.
  Start with $x \in \mathcal{V}_2$ and apply the above maps to obtain
  $\mathcal{L} = \varphi_{42}(x)\in\mathcal{V}_4$,
  $\tilde{x} = \varphi_{34}(\mathcal{L})\in\mathcal{V}_3$,
  $x' = \varphi_{23}(\tilde{x})\in\mathcal{V}_2$.
  From \cref{eq:2to4,eq:4to3} we have:
  \bes
    \label{eq:3to2.proof}
  \begin{align}
    \tilde{x}_{ijkl}
    &= x_{ijkl}
      - \frac12 \sum_{n=1}^d (x_{njin}\delta_{lk} + x_{nlkn}\delta_{ij}) \\
    \label{eq:3to2.proof:b}
    &= x_{ijkl} + \tilde{b}_{ij} \delta_{kl}
      + \delta_{ij} \tilde{b}_{kl}\ ,
  \end{align}
  \ees
  where
  \begin{equation}
    \label{eq:3to2.proof.b}
    \tilde{b}_{ij} = -\frac12 \sum_{n=1}^{d} x_{njin}.
  \end{equation}
  In order to compute $x'=\varphi_{23}(\tilde{x})$ we compute $b$
  given by \cref{eq:3to2.b} by substituting
  \cref{eq:3to2.proof:b}
 into Eq.~\eqref{eq:3to2.b} and using Eq.~\eqref{eq:xcond}.
  \bes
  \begin{align}
      b_{ij} &= \frac{1}{2d} \sum_{k=1}^d \biggl(
        x_{ijkk} + \tilde{b}_{ij} \delta_{kk}
      + \delta_{ij} \tilde{b}_{kk} \\
      &\quad+ x_{kkij} + \tilde{b}_{kk} \delta_{ij}
      + \delta_{kk} \tilde{b}_{ij}
      -\frac{1}{d} \sum_{l=1}^d(x_{kkll} + 2\tilde{b}_{kk}\delta_{ll})\delta_{ij}
      \biggr) \notag{} \\
      &= \tilde{b}_{ij}
        + (1-1)\frac{1}{d} \delta_{ij} \sum_{k=1}^d \tilde{b}_{kk}
      = \tilde{b}_{ij}\ .
  \end{align}
  \ees
  Therefore, Eq.~\eqref{eq:3to2} subtracts the same term as the one added
  in Eq.~\eqref{eq:3to2.proof}. Hence, $x = x'$. Note that the
  computation we performed also implies that $\varphi_{23}$
  is well-defined by a similar argument as used for $\varphi_{42}$ above.

  Finally, concerning $\varphi_{13}$ and $\varphi_{43}$, one can check that the term of the form
  \begin{equation}
  \label{eq:b-term}
    \tilde{b}_{ij} \delta_{kl} + \delta_{ij} \tilde{b}_{kl} = {b}_{ij} \delta_{kl} + \delta_{ij} {b}_{kl}\ ,
  \end{equation}
  which appears in Eq.~\eqref{eq:3to2}, evaluates to $0$ when substituted 
  into Eqs.~\eqref{eq:2to1.H}, \eqref{eq:2to1.a}, and~\eqref{eq:2to4}. Specifically, 
      \bes
      \begin{align}
         H_{ij} &= \frac1{2id}\sum_{k=1}^{d} [(\tilde{x}_{kkij}
          - b_{kk} \delta_{ij} - b_{ij}\delta_{kk})\notag \\
          &\qquad -(\tilde{x}_{ijkk}
          - b_{ij}\delta_{kk} - \delta_{ij} b_{kk} )] \\
          \label{eq:3to1.H}
         & = \frac1{2id}\sum_{k=1}^{d} (\tilde{x}_{kkij} - \tilde{x}_{ijkk})\ .
      \end{align}
      \ees
Substituting just the term ${b}_{ij} \delta_{kl} + \delta_{ij} {b}_{kl}$ into Eq.~\eqref{eq:2to1.a}, we have:
      \bes
      \begin{align}
              &\sum_{i,j,k,l=1}^{d} (F_{m})_{ji} (b_{ij} \delta_{kl} + \delta_{ij} b_{kl}) (F_{n})_{lk} \\
          &\quad =  \sum_{i,j=1}^{d} b_{ij} [(F_m)_{ji} \Tr(F_n) + (F_n)_{ji} \Tr(F_m)] \\
          &\quad = 0 \ ,
           \end{align}
      \ees
      which means that
            \begin{equation}
        a_{mn} = \sum_{i,j,k,l=1}^{d} (F_{m})_{ji} \tilde{x}_{ijkl} (F_{n})_{lk}\ .
         \label{eq:3to1.a}
      \end{equation}
Likewise, substituting the same Eq.~\eqref{eq:b-term} term into Eq.~\eqref{eq:2to4}, we have:
\bes
     \begin{align}
& \sum_{j,k=1}^{d}
        [b_{ij}\delta_{kl}X_{jk} +\delta_{ij} b_{kl} X_{jk} \notag\\
        &\qquad - \frac12 (b_{jk}\delta_{ij} X_{kl} + \delta_{jk}b_{ij}X_{kl}) \notag\\
        &\qquad - \frac12 (X_{ij}b_{kl}\delta_{jk}
        + X_{ij}\delta_{kl}b_{jk}) ] \\
        &\quad = \frac12 \sum_{j=1}^{d} (
        2b_{ij}X_{jl} + 2b_{jl} X_{ij}
        - b_{ij} X_{jl} - b_{ij}X_{jl} \\
        &\qquad - X_{ij}b_{jl} - X_{ij}b_{jl}
        ) \notag \\
        &\quad = 0\ ,
      \end{align}
      \ees
      which means that 
            \begin{equation}
        [\mathcal{L}(X)]_{il} = \sum_{j,k=1}^{d}
        (\tilde{x}_{ijkl} X_{jk} - \frac12 \tilde{x}_{jkij} X_{kl} -
        \frac12 X_{ij} \tilde{x}_{kljk})\ .
      \end{equation}
  This proves the remaining statements
  $\varphi_{43} \circ \varphi_{34} = \id_{\mathcal{V}_4}$ and
  $\varphi_{13} \circ \varphi_{34} \circ \varphi_{41} = \id_{\mathcal{V}_1}$. 
\end{proof}

\subsection{Explicit formulas for
\texorpdfstring{$H$ and $a$ from $G$ and $\vec{c}$}{%
H and a from G and c}}
We are now finally prepared to complete the solution of the inverse problem and provide explicit formulas for $H$ and $a$ given $G$ and $\vec{c}$. We first solve the problem more generally: assume we have $(G,\vec{c})\in \mathcal{V}_6$ and would like to know
the corresponding object from one of $\mathcal{V}_1,\dots,\mathcal{V}_5$.
If we are interested in $\mathcal{L} \in \mathcal{V}_4$ or $\mathcal{L} \in \mathcal{V}_5$ we could directly apply Eq.~\eqref{eq:6to4} from \cref{thm:forward.inverse}. For $\mathcal{V}_1$, $\mathcal{V}_2$, or $\mathcal{V}_3$ we
would, however, have to compose multiple maps from that theorem.

Here we provide explicit formulas for these compositions.
\begin{mylemma}
  \label{lm:GctoHa}
  Suppose we have $(G,\vec{c}) \in \mathcal{V}_6$. Then the objects from $\mathcal{V}_1,\mathcal{V}_2,\mathcal{V}_3$ corresponding to $(G,\vec{c})$ in the sense of
  Theorem \ref{thm:forward.inverse} are given by the following formulas:
  \begin{itemize}
    \item $(H,a) \in \mathcal{V}_1$ are given by \cref{eq:GctoH,eq:Gctoa}.
    \item $\tilde{x} \in \mathcal{V}_3$ is given by
      \begin{equation}
        \label{eq:Gc2xt}
        \tilde{x}_{ijkl}
        = \sum_{n=1}^{J} \left(
            \sum_{m=1}^{J} G_{nm} (F_{m})_{kj} + c_{n} \delta_{kj}
          \right) (F_{n})_{il}.
      \end{equation}
    \item $x \in \mathcal{V}_2$ is given by Eq.~\eqref{eq:3to2}
      with $b$ given by
      \begin{equation}
        \label{eq:Gc2b}
        b = \frac{1}{2d}
        \sum_{n=1}^{J} \left(
          \sum_{m=1}^{J} G_{nm} (\{F_{m}, F_{n}\}-\delta_{mn}I/d) + 2c_{n} F_n
          \right).
      \end{equation}
  \end{itemize}
\end{mylemma}
\begin{proof}
To simplify the computation, we introduce
\begin{equation}
  \label{eq:tGn.def}
  \tilde{G}_n = \sum_{m=1}^J G_{nm} F_m + c_n I\ .
\end{equation}
Eq.~\eqref{eq:Gc2xt} is obtained by substituting
  Eq.~\eqref{eq:6to4} into Eq.~\eqref{eq:4to3}. Eq.~\eqref{eq:Gc2b}
  is obtained by substituting Eq.~\eqref{eq:Gc2xt} into Eq.~\eqref{eq:3to2.b}:
\bes
\begin{align}
         b_{ij} &= \frac{1}{2d}\sum_{k=1}^d (\tilde{x}_{ijkk} + \tilde{x}_{kkij})\\
         &=\frac{1}{2d}\sum_{k=1}^d \bigg[ \sum_{n=1}^{J} (\tilde{G}_n)_{kj} (F_{n})_{ik}
           + (\tilde{G}_n)_{ik} (F_{n})_{kj} \notag \\
          &\qquad -\frac{1}{d}\delta_{ij}
          (\tilde G_n)_{lk} (F_{n})_{kl}
          \bigg] \\
          &= \frac{1}{2d} \bigg[ \sum_{m,n=1}^{J} G_{nm}\Bigl[
            (F_n F_m)_{ij} + (F_m F_n)_{ij} \notag \\
          &\qquad - \Tr(F_m F_n)\delta_{ij}/d\Bigr]
           + \sum_{n=1}^J 2 c_n (F_n)_{ij}\bigg]\ .
\end{align}
\ees
 Eqs.~\eqref{eq:GctoH} and~\eqref{eq:Gctoa}
  are obtained by substituting $\tilde{x}$ from Eq.~\eqref{eq:Gc2xt}
  into Eq.~\eqref{eq:3to1.H} and Eq.~\eqref{eq:3to1.a} respectively:
\bes
\begin{align}
H_{ij} &= \frac{1}{2id}\sum_{k=1}^d\sum_{n=1}^J\bigg[(\tilde{G}_n)_{ik}(F_n)_{kj}  -(\tilde{G}_n)_{kj}(F_n)_{ik}\bigg] \\
&= \frac{1}{2id}\sum_{n=1}^J([\tilde{G}_n,F_n])_{ij}
= \frac1{2id}\sum_{m,n=1}^{J} G_{nm} (\left[F_m,F_n\right])_{ij} ,
\end{align}
\ees  
and
\bes
\begin{align}
a_{mn} &= \;\sum_{{i,j,k,l=1}}^{d}\; (F_m)_{ji}
  \left(\sum_{n'=1}^{J}(\tilde{G}_{n'})_{kj} (F_{n'})_{il}\right)(F_n)_{lk} \\
& = \sum_{n'=1}^{J} \Tr(\tilde{G}_{n'}F_m F_{n'}F_n) \\
& = \sum_{i=1}^{J} \Tr\left[\left(\sum_{j=1}^J G_{ij} F_j + c_i I\right)F_m F_{i}F_n\right]\ .
\end{align}
\ees  
\end{proof}

\cref{eq:GctoH,eq:Gctoa} of \cref{thm:main-result} are now a special case (the first item) of \cref{lm:GctoHa}.

Once $a$ is computed from Eq.~\eqref{eq:Gctoa}, one can check if the given ODE 
generates a Lindblad equation (and not just a Markovian quantum master equation)
by testing whether $a$ is positive semi-definite. What 
remains is to formulate a complete-positivity condition directly in terms of $G$ and 
$\vec{c}$; we do this next.

\subsection{Complete positivity}
We now complete the proof of \cref{thm:main-result}.
\begin{mylemma}[Positive-semidefiniteness condition of \cref{thm:main-result}]
\label{lm:thm:main-result:part3}
A pair $(G,\vec{c})$ generates a (completely positive) Lindblad master equation iff \cref{eq:Gc.positive} holds for all traceless $B \in \mathcal{B}(\mathcal{H})$.
\end{mylemma}
\begin{proof}
  The condition $a\geq 0$ is equivalent to the condition
  that for any vector $b=\{b_m\}_{m=1}^{J}$
  one has
  \begin{equation}
    \label{eq:cp.bab}
    \sum_{m,n=1}^{J} b_m^*a_{mn} b_n \geq 0\ .
  \end{equation}
  Such vectors are in one-to-one correspondence with $B\in M_0(d,\mathbb{C})$
  [i.e., with traceless $B \in B(\mathcal{H})$]
  via expansion in the basis $\{F_m\}$:
  \begin{equation}
    \label{eq:cp.b}
    B = \sum_{m=1}^{J} b_m F_m, \qquad b_m = \Tr(F_m B).
  \end{equation}
  Substituting \cref{eq:Gctoa} into the right
  hand side of \cref{eq:cp.bab} and simplifying it using
  \cref{eq:cp.b}
  we obtain the right-hand side of \cref{eq:Gc.positive}:
  \begin{align}
 & \sum_{m,n=1}^{J} b_m^*a_{mn} b_n = \notag{}\\
 & \quad \sum_{m,n=1}^{J} b_m^*\sum_{i=1}^{J} \Tr\bigg[\bigg(
            \sum_{j=1}^{J} G_{ij} F_{j} + c_{i} I
          \bigg)F_{m}F_{i}F_{n}\bigg] b_n = \notag{}\\
 & \quad \sum_{i=1}^{J} \Tr\bigg[\bigg(
            \sum_{j=1}^{J} G_{ij} F_{j} + c_{i} I
          \bigg)B^\dag F_{i}B\bigg]\ .
  \end{align}
\end{proof}

\subsection{Complete positivity and convex geometry}
There is a fruitful complementary description of our complete positivity results using convex geometry, which we explain in this subsection. For the necessary background in convex geometry, see \cref{sec:convex}.

The last part of \cref{thm:main-result},
i.e., the just-proven \cref{lm:thm:main-result:part3},
describes the set of pairs $(G, \vec{c})$ corresponding to
positive semidefinite $a$ (or, equivalently, to
Lindblad master equations). We denote this set by $\mathcal{V}_{6{+}}$
(a subset of $\mathcal{V}_6$). Moreover, we
introduce $\mathcal{V}_{i{+}}$ ($i=1,\dots,6$) to be images of
$\mathcal{V}_{6{+}}$
under the maps $\varphi_{i6}$ described by
the commutative diagram \cref{fig:V1V6}.
$\mathcal{V}_{1{+}}$ consists of pairs $(H,a)$ where $H$ is an arbitrary
Hermitian matrix and $a$ is positive semidefinite.

The set $M_{+}(d,\mathbb{C})$ of positive semidefinite matrices
is a convex cone with compact support (see \cref{lm:Mplus.1} in \cref{sec:convex}).
The complete positivity part of \cref{thm:main-result} uses the 
description of the cone $M_{+}(d,\mathbb{C})$ along with a set of linear 
inequalities \cref{eq:cp.bab}, i.e., the set of supporting hyperplanes.

An alternative is to note that $M_{+}(d,\mathbb{C})$ is a convex 
hull of its extreme rays (\cref{lm:Mplus.3}). 
The extreme rays of $M_{{+}}(J,\mathbb{C})$ are the rays generated by rank 1 
matrices (\cref{lm:Mplus.2}), i.e., matrices $a$ of the form $bb^\dag$, where $b\in\mathbb{C}^J$.
  
  Using the map $\varphi_{61}$ this provides a description for 
$\mathcal{V}_{6{+}}$ given by the following:
\begin{myproposition}
  The cone $\mathcal{V}_{6+}$ of pairs $(G,\vec{c})$ generating a
  (completely positive) Lindblad master equation coincides with the convex hull of the 
  following elements:
  \begin{itemize}
    \item Elements of the form $(Q, \vec{0})$, where $Q$ is given by \cref{eq:QfromH} for $H \in M_{0,\textrm{sa}}(d,\mathbb{C})$;
    \item Elements of the form $(R, \vec{c})$, where
    \bes
    \begin{align}
      R_{kl} &= \Tr\big[F_k \big(B F_l B^\dagger
        - \frac{1}{2} \left\{B^\dagger B, F_l\right\} \big)\big] \\
      c_k &= \frac{1}{d} \Tr(\big[B , B^\dagger \big]F_k)\\
      B &\in M_0(d, \mathbb{C})\ .
    \end{align}
    \ees
  \end{itemize}
\end{myproposition}
\begin{proof}
  According to the discussion above, $\mathcal{V}_{6+}$ is generated
  by elements of the form $\varphi_{61}((H,a=0))$ and
  $\varphi_{61}((H=0, a=bb^\dagger))$. The former are given by \cref{eq:QfromH}. 
  The image of the latter can be computed using
  \cref{eq:Rfroma,eq:cfroma},
  after identifying $b \in \mathbb{C}^J$ with elements $B \in M_0(d,\mathbb{C})$
  via \cref{eq:cp.b}:
  \bes
  \begin{align}
    R_{kl} &= \sum_{i,j=1}^{J} b_{i} b_{j}^*\Tr\big[F_k \big(F_i F_l F_j
      - \frac{1}{2} \left\{F_j F_i, F_l\right\} \big)\big] \notag\\
    &= \Tr\big[F_k \big(B F_l B^\dagger
      - \frac{1}{2} \left\{B^\dagger B, F_l\right\} \big)\big] \\
      c_k &= \frac{1}{d} \sum_{i,j=1}^{J} b_{i} b_{j}^*\Tr(\big[F_i , F_j \big]F_k)
    = \frac{1}{d} \Tr(\big[B , B^\dagger \big]F_k)\ .
  \end{align}
  \ees

\end{proof}

\subsection{Consistency}
\label{ss:consistency}

By construction, the forward maps defined in \cref{thm:forward.inverse} are 
consistent with our previously defined maps. More explicitly, this is stated
in the following proposition.
\begin{myproposition}
\label{prop:GtoHconsistency}
  Let $H'$ be a Hermitian operator in $\mc{B}(\mathcal{H})$ and let $a$ be
  a Hermitian $J\times J$ matrix. Let $Q,R,\vec{c}$
  be obtained from $(H',a)$
  using \cref{eq:Qjk,eq:Rjk,eq:cjdef}, and $G=Q+R$. Let $H=H'-\frac{1}{d}\Tr(H')I$ be the traceless component of $H'$, and let $\varphi_{ij}$ be the maps from \cref{thm:forward.inverse}. Then
  \bes
  \begin{align}
    \label{eq:Q0eqphi61}
    (Q,0) &= \varphi_{61}((H,0)), \\
    \label{eq:Rceqphi61}
    (R,c) &= \varphi_{61}((0,a)), \\
    \label{eq:Gceqphi61}
    (G,c) &= \varphi_{61}((H,a)).
  \end{align}
  \ees
\end{myproposition}
\begin{proof}
  First, note that \cref{eq:L_H} yields the same
  operator $\mathcal{L}$ if we change $H$ by a constant.
  Hence, we could apply \cref{eq:Qjk} to $H$ instead of $H'$
  and get the same $Q$. One of the equivalent ways to
  describe $\varphi_{16}$ is $\varphi_{61} = \varphi_{65} \circ \varphi_{54} \circ \varphi_{41}$. Note that
  \begin{equation}
  (\varphi_{54} \circ \varphi_{41})((H,0)) =
  \mathcal{L}_{H} +
  \left.\mathcal{L}_{a}\right|_{a=0}
  = \mathcal{L}_{H}.
  \end{equation}
  Thus, \cref{eq:Q0eqphi61} follows from comparison of \cref{eq:Qjk} with \cref{eq:5to6}.
  Similarly,
  \begin{equation}
  (\varphi_{54} \circ \varphi_{41})((0,a)) = \mathcal{L}_{a},
  \end{equation}
  and \cref{eq:Rceqphi61} follows from comparison of
  \cref{eq:Rjk,eq:cjdef} with \cref{eq:5to6}.
  Finally, \cref{eq:Gceqphi61} is the sum of \cref{eq:Q0eqphi61,eq:Rceqphi61}.
\end{proof}

In the following Proposition, we show that \cref{eq:GctoH,eq:HfromQ} are consistent, i.e., yield the same result when used to recover $H$. 

\begin{myproposition}
\label{prop:recoverHconsistency1}
For any $J\times J$ matrix $G$, the r.h.s. of \cref{eq:GctoH} and \cref{eq:HfromQ} (when $Q$ is replaced with $G$) are equal
to each other.
\end{myproposition}
\begin{proof}
  Comparing \cref{eq:GctoH} to \cref{eq:HfromQ} one can see that for a given $G$ the claim is equivalent to
  \begin{equation}
    \label{eq:Hconsist:1}
    \frac{1}{2id} G_{nm} [F_m,F_n] = -\frac{1}{2d} f_{nml} G_{nm} F_{l}
  \end{equation}
  where we use the Einstein summation notation. \cref{eq:Hconsist:1} follows directly from the definition of the structure constants \cref{eq:F-Lie}.
\end{proof}

Next, we show that the inverse maps [not only \cref{eq:GctoH,eq:HfromQ}] also
yield the same result when used to recover $H$,
and in fact one may interchange $G$ and $Q$ when using these
maps, or subtract any matrix from $G$ which can be obtained using the formula for $R$ (even from a different $\mathcal{L}_a$; e.g., any real symmetric matrix).

\begin{myproposition}
  \label{prop:recoverHconsistency2}
  Let $H,a,Q,R,G,\vec{c}$ be the same as in \cref{prop:GtoHconsistency}. Then all of the following methods recover the same $H$:
  \begin{enumerate}
    \item Applying \cref{eq:GctoH} to $G$ (as is);
    \item Applying \cref{eq:GctoH} to $Q$ in place of $G$;
    \item Computing $Q'$ to be the antisymmetric part of $G$
    and applying \cref{eq:GctoH} to $Q'$ in place of $G$;
    \item Computing $Q'' = G - R'$ where $(R',\vec{c}') = \varphi_{61}((0,a'))$ for any Hermitian $J\times J$ matrix $a'$ (possibly different from $a$) and applying \cref{eq:GctoH} to $Q''$ in place of $G$.
  \end{enumerate}
\end{myproposition}
\begin{proof}
    According to \cref{thm:forward.inverse}, the diagram \cref{fig:V1V6} is commutative. Since \cref{eq:GctoHa} is
    is the formula for $\varphi_{16}$ and $\varphi_{61}((H,a)) = (G,c)$, applying \cref{eq:GctoH} to $G$ recovers the original $H$ (statement \#1). 
    
    Since $\varphi_{61}((H,0)) = (Q,\vec{0})$,
    applying \cref{eq:GctoH} to $Q$ also recovers the original $H$ (statement \#2).

    Since \cref{eq:GctoH} contracts $G_{nm}$ with a commutator,
    which is antisymmetric with respect to the order of indices,
    the result of that application is independent of adding
    any symmetric matrix to $G$ (statement \#3).

    The final statement follows from
    $\varphi_{61}((H, a - a')) = (Q'',c-c')$.
\end{proof}

\subsection{Decomposing
\texorpdfstring{$G$ into $G = Q + R$}{G into G = Q + R}, and the
space of possible matrices \texorpdfstring{$R$}{R}}
\label{ss:Decomposing}

\begin{myproposition}
\label{prop:Decomposing}
The decomposition of $G$ into $Q$ and $R$ can be obtained by
the following formulas:
\bes
\label{eq:QRfromG}
\begin{align}
  \label{eq:QfromG.F}
  Q_{ij} &= -\frac1{2d}\sum_{m,n=1}^{J} G_{nm}\Tr\left(
  [F_{i},F_{j}],[F_{n},F_{m}]\right) \\
  \label{eq:QfromG.f}
  &= \frac1{2d}\sum_{m,n,k=1}^{J} G_{nm}f_{knm}f_{kij}.\\
  \label{eq:RfromG}
  R_{ij} &= G_{ij} - Q_{ij}.
\end{align}
\ees
\end{myproposition}
\begin{proof}
Substituting \cref{eq:GctoH} into \cref{eq:QfromH-a}, we obtain:
\begin{equation}
  Q_{ij} = \frac1{2d}\sum_{m,n=1}^{J} G_{nm}\Tr\left(
    F_{i}\left[F_{j},\left[F_m,F_n\right]\right]\right),
\end{equation}
which is equivalent to \cref{eq:QfromG.F}.
\cref{eq:QfromG.f} is then obtained
using the definition of the structure constants in \cref{eq:F-Lie}. Finally, \cref{eq:RfromG} follows from $G = R + Q$.

Note that an alternative way to obtain the formula
for $R$ is to substitute \cref{eq:Gctoa} into \cref{eq:Rfroma}
(which results in a significantly more complex computation
giving the same final result).
\end{proof}

Naively, one might be tempted to set $Q$
to be the antisymmetric part of $G$, or $R$
to be the symmetric part of $G$, instead of using the expressions given in \cref{prop:Decomposing}. This, however,
only works for $d \leq 2$.
This follows
from the dimension counting done in \cref{cor:Rdim} below: the space of antisymmetric matrices $R$
has dimension $J(J-3)/2$, which is non-vanishing unless $d\in\{1,2\}$ (recall that $J=d^2-1$).
We give an example illustrating this for $d=3$ in \cref{sss:factor-su3}.

However, as stated in 
\cref{prop:recoverHconsistency1,prop:recoverHconsistency2}
above, taking the antisymmetric part of $G$ would
still result in the correct $H$ being recovered.

\begin{myproposition}
\label{prop:Rcond}
For any $R \in M(J,\mathbb{R})$, $\vec{c} \in \mathbb{R}^J$
a pair $(R, \vec{c})$ can be obtained from some Hermitian
matrix $a$ if and only if
\begin{equation}
  \label{eq:Rcond}
  \sum_{m,n=1}^{J} R_{mn} [F_n, F_m] = 0.
\end{equation}
In particular, this condition is independent of $\vec{c}$,
i.e., for any $\vec{c}' \in \mathbb{R}^J$ the
pair $(R, \vec{c})$ can be obtained
from some Hermitian matrix $a$ if and only if $(R, \vec{c}')$ can
(from a possibly different Hermitian matrix $a'$).
\end{myproposition}
\begin{proof}
According to \cref{prop:GtoHconsistency} a pair $(R, \vec{c})$
can be obtained from some Hermitian matrix $a$ if and only if
$(R,\vec{c}) \in \varphi_{61}(\{(0,a): a \in M_{\textrm{sa}}(J,\mathbb{C})\})$.
According to \cref{thm:forward.inverse} this is equivalent
to
\begin{equation}
  H=0\quad\textrm{ where }\quad(H, a) = \varphi_{16}((R,\vec{c})).
\end{equation}
From \cref{lm:GctoHa},
this is equivalent to \cref{eq:GctoH}
evaluating to 0 when $R$ is used instead of $G$, i.e.
when \cref{eq:Rcond} holds.
\end{proof}

\begin{mycorollary}
    \label{cor:Rdim}
    The space of matrices $R$ which can be obtained
    from some Hermitian matrix $a$ is an $\mathbb{R}$-linear
    subspace of $M(J,\mathbb{R})$ of dimension $J^2-J$
    and includes the space of symmetric matrices.

    The space of antisymmetric matrices $R$, which can
    be obtained from some Hermitian matrix $a$ has
    dimension $J(J-3)/2$.
\end{mycorollary}
\begin{proof}
  As explained in the proof of \cref{prop:Rcond}, the space
  of possible $(R,c)$ is a bijective image of a $J^2$ dimensional space. According to \cref{prop:Rcond} it contains
  all possible pairs of the form $(0,\vec{c})$ ---- a $J$-dimensional subspace, hence the space of possible values for $R$ has dimension $J^2-J$.

  Furthermore, any symmetric matrix satisfies \cref{eq:Rcond}
  due to contraction with an antisymmetric commutator there.

  The final statement follows from the direct calculation:
  \begin{equation}
    J^2 - J - J(J+1)/2 = J(J-3)/2\ ,
  \end{equation}
  where $J(J+1)/2$ is the dimension of the space of symmetric matrices.
\end{proof}

\section{Examples}
\label{sec:examples}

This section provides several examples to illustrate the general theory developed above. Various supporting calculations can be found in Ref.~\cite{ODE-to-GKLS-supporting-calcs}.

\subsection{A qubit}

Consider a single qubit, i.e., $\mc{H}=\mathbb{C}^2$. As the nice operator basis, we choose the normalized Pauli matrices $\frac{1}{\sqrt{2}}\{I,\sigma^x,\sigma^y,\sigma^z\}$.

\subsubsection{A qubit with pure dephasing}

Consider the Lindblad equation for pure dephasing: $\dot{\rho} = \mathcal{L}_a(\rho) = \g(\sigma^z\rho \sigma^z-\rho)$, for which
\begin{align}
\label{eq:a-qubit-dephasing}
    a = \begin{bmatrix}
    0 &0 &0  \\ 
    0 &0 &0\\ 
    0& 0& 2\gamma \\ 
    \end{bmatrix}\ .
\end{align}
Here $Q=0$ (since $H=0$) and, by
\cref{prop:calL-Hermit,prop:a-symm-L-unital,prop:5},
$\vec{c}=0$, since $a$ is real-symmetric. Using \cref{eq:R} we find
\begin{align}
    G = R = \begin{bmatrix}
    -2\gamma &0 &0  \\ 
    0 &-2\gamma &0\\ 
    0& 0& 0 \\ 
    \end{bmatrix} \ .
\end{align}

The inverse problem retrieves $H$ from $G$ and $a$ from $(G,\vec{c})$. 
Using \cref{eq:GctoH,eq:Gctoa}, we indeed find $H=0$ and $a$ as given in \cref{eq:a-qubit-dephasing}.

\subsubsection{A qubit with amplitude damping}
\label{sss:amplitude}
Now consider the Lindblad equation for spontaneous emission: 
\bes
\begin{align}
    \dot{\rho} &= \mathcal{L}_H(\rho) + \mathcal{L}_a(\rho) \\
    \mathcal{L}_H(\rho) &= -i \o [\sigma^z,\rho ] \\
    \mathcal{L}_a(\rho) &= \g(\sigma^-\rho \sigma^+ -\frac{1}{2}\{\sigma^+\sigma^-,\rho\})\ ,
\end{align}
\ees
where $\sigma^\pm = \frac{1}{2}(\sigma^x\mp i \sigma^y)$.
This yields
\begin{align}
\label{eq:a-qubit-AD}
a = \left[
\begin{array}{ccc}
 \gamma  & -i \gamma  & 0 \\
 i \gamma  & \gamma  & 0 \\
 0 & 0 & 0\\
\end{array}
\right]\ ,
\end{align}
i.e., $a$ is not real-symmetric, and the Lindbladian is non-unital (in particular, $\vec{c}$ is non-zero). 

Then, using \cref{eq:QfromH,eq:R} we find
\bes
\label{amplitudeR}
\begin{align}
Q &= \left[
\begin{array}{ccc}
 0 & -2 \omega  & 0 \\
 2 \omega  & 0 & 0 \\
 0 & 0 & 0 \\
\end{array}
\right] \\
    R &= \begin{bmatrix}
 -\gamma & 0 & 0 \\
 0 & -\gamma & 0 \\
 0 & 0 & -2\gamma
    \end{bmatrix}\ ,
\end{align}
\ees
and using \cref{eq:cfroma}
\begin{align}
\label{amplitudec}
    \vec{c} = (0,0,\sqrt{2}\gamma)^T\ .
\end{align}
Using \cref{eq:GctoH,eq:Gctoa} with $G=Q+R$ we then indeed find $H=\o \sigma^z$ and and $a$ as given in \cref{eq:a-qubit-AD}.

\subsubsection{Example of \texorpdfstring{$G$}{G} giving rise to a non-CP map}
\label{sss:notpotpp}

Let $\vec{c}=\vec{0}$ and
\begin{align}
\label{eq:G-qubit-nonCP}
    G = \begin{bmatrix}
    0 &1 &0  \\ 
    0 &0 &0\\ 
    0& 0& 0\\ 
    \end{bmatrix} \ ,
\end{align}
Using \cref{eq:Gctoa} this yields:
\beq
a=\left[
\begin{array}{ccc}
 0 & \frac{1}{2} & 0 \\
 \frac{1}{2} & 0 & 0 \\
 0 & 0 & 0 \\
\end{array}
\right]\ ,
\eeq
whose eigenvalues are $\{-1/2,1/2,0\}$. I.e., $a$ is not
positive semidefinite, which means that the corresponding Markovian quantum
master equation does not generate a CP map. Hence, the pair $(G,\vec{c})$
with matrix $G$ given in
\cref{eq:G-qubit-nonCP} and $\vec{c}=0$ does not correspond to a Lindblad equation,
only to a Markovian quantum master equation.

\subsection{A qutrit}
\label{ss:su3-example}
As the nice operator basis for this $d=3$-dimensional case, we choose the $8$ normalized Gell-Mann matrices, including the normalized identity matrix:
\begin{align}
&F_0=\frac{1}{\sqrt{3}}I,\ \{F_m\}_{m=1}^d=\left\{
\left(
\begin{array}{ccc}
 0 & \frac{1}{\sqrt{2}} & 0 \\
 \frac{1}{\sqrt{2}} & 0 & 0 \\
 0 & 0 & 0 \\
\end{array}
\right),
\left(
\begin{array}{ccc}
 0 & 0 & \frac{1}{\sqrt{2}} \\
 0 & 0 & 0 \\
 \frac{1}{\sqrt{2}} & 0 & 0 \\
\end{array}
\right), \right. \notag  \\
&\left. \left(
\begin{array}{ccc}
 0 & 0 & 0 \\
 0 & 0 & \frac{1}{\sqrt{2}} \\
 0 & \frac{1}{\sqrt{2}} & 0 \\
\end{array}
\right),
\left(
\begin{array}{ccc}
 0 & -\frac{i}{\sqrt{2}} & 0 \\
 \frac{i}{\sqrt{2}} & 0 & 0 \\
 0 & 0 & 0 \\
\end{array}
\right),
\left(
\begin{array}{ccc}
 0 & 0 & -\frac{i}{\sqrt{2}} \\
 0 & 0 & 0 \\
 \frac{i}{\sqrt{2}} & 0 & 0 \\
\end{array}
\right),\right. \notag \\
&\left. \left(
\begin{array}{ccc}
 0 & 0 & 0 \\
 0 & 0 & -\frac{i}{\sqrt{2}} \\
 0 & \frac{i}{\sqrt{2}} & 0 \\
\end{array}
\right),\left(
\begin{array}{ccc}
 \frac{1}{\sqrt{2}} & 0 & 0 \\
 0 & -\frac{1}{\sqrt{2}} & 0 \\
 0 & 0 & 0 \\
\end{array}
\right),\left(
\begin{array}{ccc}
 \frac{1}{\sqrt{6}} & 0 & 0 \\
 0 & \frac{1}{\sqrt{6}} & 0 \\
 0 & 0 & -\sqrt{\frac{2}{3}} \\
\end{array}
\right) \right\}
\label{eq:su(3)-matrices}
\end{align}

\subsubsection{A qutrit with amplitude damping}
\label{sss:qutrit-amplitude-damping}
Consider the example of $a$ given in \cref{eq:a-su3}, i.e., a qutrit undergoing amplitude damping between its two lowest levels (similar to the qubit example given in \cref{sss:amplitude}). The corresponding $R$ matrix was given in \cref{eq:R-su3}; assuming $H=0$ we have $G=R$. In addition, we find 
\beq
\vec{c}=(0,0,0,0,0,\frac{\sqrt{2}}{3},0,0)\ .
\label{eq:c3}
\eeq
Computing $a$ using \cref{eq:Gctoa}, we find that it is indeed identical to \cref{eq:a-su3}.

\subsubsection{Illustration of \texorpdfstring{\cref{prop:Decomposing,prop:recoverHconsistency2}}{Prop. \ref*{prop:Decomposing}, \ref*{prop:recoverHconsistency2}}}
\label{sss:factor-su3}

Note that for symmetric $G$, the decomposition is trivial: $Q = 0$, $R = G$.
Therefore, to illustrate \cref{prop:Decomposing},
let us choose an arbitrary antisymmetric $G$; for example
\begin{align}
	\label{eq:Gforsu3}
    G = \left(
\begin{array}{cccccccc}
 0 & 0 & 0 & 0 & -\frac{1}{4} & 0 & 0 & 0 \\
 0 & 0 & 0 & \frac{1}{4} & 0 & 0 & 0 & 0 \\
 0 & 0 & 0 & 0 & 0 & 0 & 0 & 0 \\
 0 & -\frac{1}{4} & 0 & 0 & 0 & 0 & 0 & 0 \\
 \frac{1}{4} & 0 & 0 & 0 & 0 & 0 & 0 & 0 \\
 0 & 0 & 0 & 0 & 0 & 0 & \frac{1}{4} & -\frac{\sqrt{3}}{4} \\
 0 & 0 & 0 & 0 & 0 & -\frac{1}{4} & 0 & \frac{\sqrt{3}}{4} \\
 0 & 0 & 0 & 0 & 0 & \frac{\sqrt{3}}{4} & -\frac{\sqrt{3}}{4} & 0 \\
\end{array}
\right)\ .
\end{align}
Computing the Hamiltonian $H$ that arises from this $G$ using \cref{eq:HfromQ} [or, equivalently, \cref{eq:GctoH}] yields:
\beq
\label{eq:HfromQsu3}
H = \left(
\begin{array}{ccc}
	0 & 0 & 0 \\
	0 & 0 & \frac{1}{6} \\
	0 & \frac{1}{6} & 0 \\
\end{array}
\right)\ .
\eeq
When computing $Q$ from this Hamiltonian using \cref{eq:QfromH}, we observe
that 
$Q$ does not equal the antisymmetric part of $G$, as anticipated by \cref{cor:Rdim} and the comment after \cref{prop:Decomposing}. Instead, we find:
\begin{align}
\label{eq:Q-su3}
	Q = \left(
	\begin{array}{cccccccc}
		0 & 0 & 0 & 0 & \frac{1}{6} & 0 & 0 & 0 \\
		0 & 0 & 0 & \frac{1}{6} & 0 & 0 & 0 & 0 \\
		0 & 0 & 0 & 0 & 0 & 0 & 0 & 0 \\
		0 & -\frac{1}{6} & 0 & 0 & 0 & 0 & 0 & 0 \\
		-\frac{1}{6} & 0 & 0 & 0 & 0 & 0 & 0 & 0 \\
		0 & 0 & 0 & 0 & 0 & 0 & \frac{1}{6} & -\frac{1}{2 \sqrt{3}} \\
		0 & 0 & 0 & 0 & 0 & -\frac{1}{6} & 0 & 0 \\
		0 & 0 & 0 & 0 & 0 & \frac{1}{2 \sqrt{3}} & 0 & 0 \\
	\end{array}
	\right)\ .
\end{align}
As expected, if we then use \cref{eq:QfromG.F} to decompose $G$ into $Q$ and $R$, we get the same $Q$ as the one in \cref{eq:Q-su3}. Also, as guaranteed by \cref{prop:recoverHconsistency2}, if we use that $Q$ to compute $H$ again,
we will get the same $H$ as in \cref{eq:HfromQsu3}.

\section{The rarity of Lindbladians}
\label{sec:rarity}

As posed in the title, the question that motivated this work can be interpreted as the probability 
that a given pair $(G,\vec{c})$ will give rise to a Lindbladian (given a natural choice of a 
distribution over the pairs $(G,\vec{c})$). Alternatively, one can use a natural choice of a
distribution for $(H, a)$ in order to induce the distribution on the pairs $(G,\vec{c})$ and
compute a probability that a given pair $(G,\vec{c})$ with that distribution
gives rise to a Lindbladian. In the remainder of this section, we provide partial answers
to these questions and describe the difference between these two distributions.

\subsection{Natural distribution on the pairs of \texorpdfstring{$(G,\vec{c})$}{(G,c)}}
\label{ssec.Gc.rarity}
Suppose all elements of $G$ and $\sqrt{d}\vec{c}$ are picked 
independently from a standard normal distribution. Such a distribution 
is known as the Ginibre Orthogonal Ensemble (GinOE)~\cite{GinOE-GinSE}.
One might be interested
in the following question: ``How likely is it that these $G$ and $\vec{c}$
generate a Lindbladian, i.e. $a\geq 0$?''. Let us denote this probability
as $\tilde{P}^{\text{GinOE}}_J$. While we do not know the asymptotics
of $\tilde{P}^{\text{GinOE}}_J$, here we provide a non-rigorous attempt
at deriving the asymptotics of an upper bound. Recall that a necessary condition for $a\geq 0$ is
that all eigenvalues of $G$ have a non-positive real part.
Denoting the probability of the latter by $P^{\text{GinOE}}_J$, this implies
\begin{equation}
  \label{eq:tPleqP}
  \tilde{P}^{\text{GinOE}}_J \leq P^{\text{GinOE}}_J.
\end{equation}

The distribution of eigenvalues in the GinOE is known: see, e.g. \cite[Eq.~(1.7)]{GinOE-GinSE}. One can use either the methods of 
\cite{dean2008extreme} or some rigorous alternative to those methods
to estimate the asymptotics of $P^{\text{GinOE}}_J$, which would be an upper bound for $\tilde{P}^{\text{GinOE}}_J$ due to \cref{eq:tPleqP}.

More specifically, assuming the same (non-rigorous) argument as in Ref.~\cite{dean2008extreme} applies to the GinOE, we can state that
\begin{equation}
  P^{\text{GinOE}}_J = \exp(-\theta^{\text{GinOE}} J^2 + O(J)).
  \label{eq:P-GinOE}
\end{equation}
We discuss the estimation of the positive constant $\theta^{\text{GinOE}}$ in \cref{app:theta}. The important point about \cref{eq:P-GinOE} is that the probability decays rapidly in the Hilbert space dimension $d$ (recall that $J=d^2-1$).

\subsection{Natural distribution on \texorpdfstring{$a$}{a}}
\label{ssec:a.rarity}
In this subsection, we examine the prevalence of Lindbladians within the set of all Liouvillians.

Specifically, we attempt to answer the following question:
``Given a random $H$ and $a$, what is the probability that $\mathcal{L}_{H,a}$ is a Lindbladian?''
Since the condition $a\geq 0$ that guarantees a Liouvillian is a Lindbladian depends only on $a$ and not on $H$,
defining the distribution in the space of Hermitian matrices $a$ is sufficient to answer this question.
Additionally, the scale of $a$ is irrelevant; changing $a$ to $\alpha a$ (where $\alpha>0$ can be randomly selected from a distribution which may depend on $a$)
does not alter the answer to $a\geq 0$.

One natural choice for the distribution is the Gaussian Unitary Ensemble (GUE), also known as the $\beta=2$ Gaussian Ensemble.
In this ensemble, the real and imaginary parts of each component of a $J\times J$ matrix $A$ are independently selected from
a normal distribution with a mean of $0$
and\footnote{We choose the normalization to be consistent with Ref.~\cite{forrester2010log}. An alternative normalization used in other works is variance $=1$, i.e., the standard normal distribution.}
a variance of $1/\beta=1/2$.
The matrix $a$ is then computed as $a=(A+A^\dagger)/2$. The condition
$a\geq 0$ is equivalent to $\lambda_j\geq 0$ for $j=1,\dots,J$,
where $\lambda=\{\lambda_j\}_{j=1}^J$ is the vector of eigenvalues of $a$. Let us denote the
probability of this event as $P^{\text{GUE}}_J$.

The joint probability density function (PDF) of eigenvalues of a matrix from the GUE is well known to be:
\bes
\begin{align}
    p(\lambda) &=\frac{1}{G_{\beta,J}}e^{-\frac{\beta}{2}\sum_{j=1}^{J} \lambda_j^2}\prod_{1\leq j < k \leq J}\!\!\!\! \abs{\lambda_k - \lambda_j}^{\beta} \\
    &= \frac{1}{G_{\beta,J}}e^{-\beta E(\lambda)}\ ,
\end{align}
\ees
where $E(\lambda)=\frac{1}{2}\sum_{j=1}^J\lambda_j^2-\sum_{1\leq j<k\leq J}\ln\abs{\lambda_j-\lambda_k}$ and $G_{\beta,J}$ is the normalization factor (partition function);
see \cite[Eq.~(1.4)]{forrester2010log} for further details.
Thus, the distribution of $\lambda$ can be interpreted as the canonical ensemble of a two-dimensional gas (in $\mathbb{C}$)
with $J$ charged particles (each of unit charge with the same sign), constrained to the real line in a harmonic potential at temperature $1/\beta = 1/2$, also known as a log-gas.
The probability $P^{\text{GUE}}_J$ in this interpretation is the probability that, by random chance, all particles happen to be to the right of the origin.
$P^{\text{GUE}}_J$ is known as the probability of atypically large fluctuations of extreme value statistics of the GUE,
and was estimated in Ref.~\cite{dean2008extreme} as:
\begin{equation}
    \label{eq:NUL1.PN}
    P^{\text{GUE}}_J = 3^{-J^2/2 + O(J)}\ .
\end{equation}
The estimate uses non-rigorous techniques from statistical mechanics, including
the replacement of particle density (sum of delta functions) with a smooth function
and the use of functional integrals and functional Fourier transforms.
While rigorously proving this estimate is left to future research, \cref{eq:NUL1.PN}
shows that the probability of a randomly picked Liouvillian being a Lindbladian 
decays rapidly with $d$, just like \cref{eq:P-GinOE} for the GinOE. For example, already for $d=4$, the estimate gives
$P^{\text{GUE}}_J \sim 10^{-54}$, demonstrating that it is extremely unlikely to find a
Lindbladian for $d=4$ by picking a random $a$ from the GUE.

\subsection{Comparison between these distributions}
Given the distribution on the pairs $(G,\vec{c})$ described in \cref{ssec.Gc.rarity},
one can derive a distribution for $a$ or
a matrix related to the tensor $x$ from \cref{thm:forward.inverse}.
$\tilde{P}^{\text{GinOE}}$ is then equal to the probability that $a \geq 0$ (or, equivalently,
$x \geq 0$) given that distribution.

Due to the linearity of $\varphi_{16}$,
the components of $a$ also
follow a multivariate Gaussian distribution with mean $0$.
Using \cref{eq:Gctoa} one can compute the covariances:
\begin{equation}
  \label{eq:GinOE:Vara}
    \mathbb{E}(a_{mn}a_{m'n'}) = \delta_{mn'} \delta_{nm'} - \frac{1}{d} \Tr(F_{m'}F_{n'}F_{m}F_{n})\ .
\end{equation}
This distribution does not match the GUE investigated in 
\cref{ssec:a.rarity}: for comparison covariances in the GUE are given by
\begin{equation}
  \label{eq:GUE:Vara}
  \mathbb{E}(a_{mn}a_{m'n'}) = \frac12 \delta_{mn'}\delta_{nm'}.
\end{equation}
The normalization difference ($1$ \textit{vs} $1/2$) is due to a difference
in conventions and is irrelevant to the question of whether $a\geq 0$.
The presence of the second term in \cref{eq:GinOE:Vara}, however,
is significant and, for $d>2$, introduces a dependence of
the probability distribution given by \cref{eq:GinOE:Vara}
on the choice of the nice operator basis.

Nevertheless, the distribution of $x$ computed from $(H=0,a)$
is independent of both the nice operator
basis and the choice of the orthonormal basis in the Hilbert space $\mathcal{H}$:
\begin{align}
  \label{eq:GinOE:Varx}
  &\mathbb{E}(x_{i j k l} x_{i' j' k' l'}) =
  \delta_{j k'} \delta_{j' k} \delta_{l i'} \delta_{l' i} \notag\\
  &+ d^{-1} (
    - \delta_{j k'} \delta_{j' i'} \delta_{l k} \delta_{l' i}
    - \delta_{j k'} \delta_{j' k} \delta_{l i} \delta_{l' i'}
    - \delta_{j i} \delta_{j' k} \delta_{l i'} \delta_{l' k'}) \notag\\
  &+ d^{-2} (
    \delta_{j k'} \delta_{j' i} \delta_{l k} \delta_{l' i'}
    + \delta_{j i} \delta_{j' k} \delta_{l k'} \delta_{l' i'} \notag\\
  &\quad+ \delta_{j k'} \delta_{j' i'} \delta_{l i} \delta_{l' k}
    + \delta_{j i'} \delta_{j' k} \delta_{l i} \delta_{l' k'}) \notag\\
  &+ d^{-3} (
    - \delta_{j k'} \delta_{j' i'} \delta_{l k} \delta_{l' i}
    - \delta_{j i} \delta_{j' k'} \delta_{l k} \delta_{l' i'}
    - \delta_{j i} \delta_{j' i'} \delta_{l k'} \delta_{l' k} \notag\\
  &\quad- \delta_{j k} \delta_{j' i'} \delta_{l i} \delta_{l' k'}
    - \delta_{j i} \delta_{j' k} \delta_{l i'} \delta_{l' k'}
    - \delta_{j i'} \delta_{j' i} \delta_{l k} \delta_{l' k'}) \notag\\
  &+ 4d^{-4} \delta_{j i} \delta_{j' i'} \delta_{l k} \delta_{l' k'}
  + Jd^{-4} \delta_{j i} \delta_{j' i'} \delta_{l k} \delta_{l' k'}.
\end{align}
This $x$ can be interpreted as a matrix with indices $(ij)$ and $(lk)$:
$x_{(ij)(lk)} = x_{ijkl}$. With this interpretation, such $x$ is positive 
semidefinite if and only if $a$ is.

While beyond the scope of this work, one can use \cref{eq:GinOE:Vara} or 
\cref{eq:GinOE:Varx} to attempt to estimate $\tilde{P}^{\text{GinOE}}_J$. One approach might be to 
first compute the
joint distribution of eigenvalues of $a$ or of $x$ similarly to
how they are computed for the GUE (see, e.g., Ref.~\cite{forrester2010log}),
i.e. by integrating out unitary symmetry. One may attempt to integrate 
out the unitaries in $U(\mathcal{H})$ first, and then integrate over the 
quotient space $U(B_0(\mathcal{H}))/U(\mathcal{H})$, which consists of 
the unitaries in $B_0(\mathcal{H})$ modulo those in $U(\mathcal{H})$,
where $U(\mathcal{H})$ denotes the manifold of unitaries acting on a 
finite-dimensional Hilbert space $\mathcal{H}$. That second integration
appears to be non-trivial because while \cref{eq:GinOE:Varx} is symmetric
with respect to $U(\mathcal{H})$, it is no longer symmetric with
respect to $U(B_0(\mathcal{H}))/U(\mathcal{H})$. Once the joint distribution
of eigenvalues is computed (or estimated), one can attempt to use it to
estimate the probability that all of them are non-negative.

\section{Overlap with prior work}
\label{sec:prior}

After this work was completed, it came to our attention that our main results could be reconstructed by combining several earlier results, in particular using Refs.~\cite[Section 3.2.2]{Breuer:book},
\cite[Section 7.1.2]{wolf2012quantum}, and~\cite{riedel-checkable-2020}. An aspect that is common to these three works and distinguishes them from ours is that they do not represent $\mathcal{L}$ by its action on the basis elements $F_i$, i.e., they do not have $(G, \vec{c})$. In this sense, their motivation is different from ours, given that our starting point is the title question ``which (nonhomogeneous linear first-order) differential equations correspond to the Lindblad equation?", formulated in terms of $(G,\vec{c})$ [\cref{eq:3}].
Nevertheless, it is possible to use these earlier results to prove
\cref{thm:main-result}: one can show that the map $\varphi_{64}$ given by
\cref{eq:5to6} is a bijective linear map with inverse given by \cref{eq:6to4}.
With this established, one can substitute \cref{eq:6to4} instead of $\mathcal{L}$ to apply the above results
to obtain \cref{thm:main-result}.

The overlap is described in more detail below. We first summarize it
in the following list, mentioning the parts of the \cref{thm:main-result}
and the overlaps we have found.
\begin{enumerate}
  \item \label{item:biject} Bijectivity of the map $\varphi_{61}$.
    \begin{itemize}
      \item Surjectivity follows from \cite[Section 3.2.2]{Breuer:book}.
      \item Injectivity follows from \cite[Section 7.1.2]{wolf2012quantum}.
    \end{itemize}
  \item \label{item:explicit.inverse} Explicit formulas [\cref{eq:GctoHa}] for the inverse ($\varphi_{16}$).
    \begin{itemize}
      \item Follows from \cite[Section 3.2.2]{Breuer:book}.
    \end{itemize}
  \item \label{item:cp} Complete positivity condition \cref{eq:Gc.positive} in terms of $G$ and $\vec{c}$.
    \begin{itemize}
      \item Follows from \cite{riedel-checkable-2020}.
    \end{itemize}
\end{enumerate}

\subsection{Overlap with %
\texorpdfstring{\cite[Section 3.2.2]{Breuer:book}}{(Breuer, 2002)}}
\cref{item:explicit.inverse} and surjectivity in \cref{item:biject} follow from \cite[Section 3.2.2]{Breuer:book}. To see this, one needs to take $V(t)$ given by
\begin{equation}
  V(t)\rho = \rho + t \mathcal{L}(\rho) + O(t^2)
\end{equation}
with $\mathcal{L}$ given by \cref{eq:6to4}, and write it in the form
\cite[Eq.~(3.50)]{Breuer:book} (see, e.g., \cref{sec:supop.FAF} on how to do that). Then one can follow 
\cite[Eqs.~(3.54)--(3.63)]{Breuer:book} to recover $a$ and $H$ and write $\mathcal{L}$ in the form \cref{eq:LE}.
After simplification, this would provide explicit formulas \cref{eq:GctoHa} for $a$ and $H$ representing that
$\mathcal{L}$. Since the procedure works for any $\mathcal{L}$ from $\mathcal{V}_4$ (i.e. any $(G,\vec{c})$), this
shows the surjectivity of $\varphi_{61}$ in \cref{item:biject}.

\subsection{Overlap with %
\texorpdfstring{\cite[Section 7.1.2]{wolf2012quantum}}{(Wolf, 2012)}}
Reference~\cite[Theorem 7.1]{wolf2012quantum} asserts
that a linear map is a Lindbladian if and only if it can be represented in
any of the four forms described in \cite[Eqs.~(7.20)-(7.23)]{wolf2012quantum}.
\cite[Eq.~(7.23)]{wolf2012quantum}
matches \cref{eq:LE}, where \cite{wolf2012quantum}
uses the matrix $C = a^T/2$ referred to as the ``Kossakowski matrix''.

Reference~\cite[Proposition 7.4]{wolf2012quantum}
establishes the uniqueness of the decomposition of the Lindbladian $\mL$ into $\mL_{H}$ and $\mL_{a}$,
with $H$ being unique up to an additive constant. The uniqueness of $a$
can be deduced from the second part of this proposition.

The complete positivity condition, \cref{eq:Gc.positive}, can be obtained from \cite{wolf2012quantum} in two steps.
First, by comparing \cite[Eq.~(7.20)]{wolf2012quantum} with
\cite[Eq.~(7.14)]{wolf2012quantum}, one may note that that the condition we need is called
``Conditional complete positivity'' in \cite{wolf2012quantum}.
Second, one may transform \cite[Proposition 7.2, Eq.~(7.15)]{wolf2012quantum}
into \cref{eq:Gc.positive}.

It should be noted that the expression of the complete positivity condition in
Ref.~\cite{wolf2012quantum} differs from the expression used in our paper,
and thus, some work would be required to derive one condition from the other.

\subsection{Overlap with
\texorpdfstring{\cite{riedel-checkable-2020}}{(Riedel, 2020)}}
The blog post \cite{riedel-checkable-2020} provides a complete positivity condition
in a form more similar to the one used in our work.
\cite{riedel-checkable-2020} defines\footnote{``PT'' stands for ``partial
transpose''.} $\mLPT$ as follows.
First, introduce a matrix notation for a superoperator $S$ acting on
$M_d(\mathbb{C})$ where $S$ is described by a matrix with elements
$S_{(nn')(mm')}$ (matrix rows and columns indexed by elements of $\{1,\dots,d\}^2$) such that:
\begin{equation}
  (S(A))_{nn'} = \sum_{m,m'=1}^{d} S_{(nn')(mm')} A_{mm'}.
\end{equation}
Then, define $S^{\textrm{PT}}$ using
\begin{equation}
  \label{eq:SPT.def1}
  S^{\textrm{PT}}_{(nm)(n'm')} = S_{(nn')(mm')}.
\end{equation}
An alternative way to describe $S^{\textrm{PT}}$ is
\begin{equation}
  \label{eq:SPT.def2}
  \Tr(A S^{\textrm{PT}}(B)) = \sum_{j,k=1}^{d} \left(AS(\ketb{j}{k})B\right)_{jk}.
\end{equation}

With this notation, the complete positivity condition becomes\footnote{In the blog post, the condition used a ``superprojector'' removing the trace. The equivalent alternative used here is to require that $B$ is traceless.}
\begin{equation}
  \label{eq:mLPT.positive}
  \forall B \in M_0(d,\mathbb{C})\quad\Tr(B^\dagger \mLPT(B)) \geq 0\ .
\end{equation}
Using \cref{eq:SPT.def2}, this can be rewritten as
\begin{equation}
  \label{eq:mL.positive.bp}
  \forall B \in M_0(d,\mathbb{C})\quad
  \sum_{j,k=1}^{d} \left(B^\dagger \mL(\ketb{j}{k})B\right)_{jk} \geq 0.
\end{equation}
This is the same as requiring \cref{eq:Gc.positive} to hold for all traceless
$B \in \mathcal{B}(\mathcal{H})$, as in \cref{thm:main-result}.

\section{Summary and Outlook}
\label{sec:conc}

A standard approach to solving a (time-independent quantum) Markovian quantum
master equation is to vectorize it and solve the corresponding nonhomogeneous
first-order linear ODE for the coherence vector. Here we posed the inverse
problem: when does such
a 1ODE of the form $\dot{\vec{v}} = G\vec{v} +\vec{c}$ correspond to a
Markovian quantum master equation? When does it correspond to a completely positive Markovian quantum master equation, i.e., a Lindblad equation? What are the
parameters, i.e., $a$ and $H$, of such master equations in terms of the parameters $G$ and $\vec{c}$ of the 1ODE? Finally, how ubiquitous are Lindbladians? 

We have shown that the answer to the first question is ``always''. We also expressed
the parameters of such Markovian quantum master equations using an expansion in a nice
operator basis, which yields explicit expressions for the Hamiltonian and the matrix $a$
of coefficients of the dissipator in terms of the parameters $(G,\vec{c})$ of the 1ODE; see \cref{thm:main-result}. In essence, this means that every 1ODE of the form $\dot{\vec{v}} = G\vec{v} +\vec{c}$ is directly representable as a Markovian quantum master equation, \cref{eq:LE}. However, complete positivity (i.e., whether the result is a Lindblad equation) is not guaranteed and must be checked on a case-by-case basis. Toward this end, we have also formulated the complete positivity condition directly in terms of $(G,\vec{c})$; see \cref{eq:Gc.positive}. This condition is equivalent to the positive semidefiniteness of $a$, which is simpler to check in practice after $a$ has been computed from $G$ and $\vec{c}$ using \cref{eq:Gctoa}. 

Our work assumed the setting of finite-dimensional Hilbert spaces. We left open for future research the problem of connecting 1ODEs to quantum master equations in the infinite-dimensional setting. We also left open the complete answer to the question of the ubiquity of Lindbladians: namely, while we have argued that the condition of positivity of $a$ makes Lindbladians extremely rare when viewed from the perspective of random matrix theory, the answer to what is the probability that $a \geq 0$ for all real $G$ (sampled from the Ginibre Orthogonal Ensemble (GinOE)~\cite{GinOE-GinSE}) such that $G$ has non-positive eigenvalues, is still open. Answering this question will quantify the probability that a randomly selected 1ODE results in a Lindblad master equation.

\acknowledgments
We thank Dr. Evgeny Mozgunov for his valuable input and helpful discussions. We also extend our appreciation to Dr. Frederik vom Ende for pointing out that \cref{eq:Ph-gen-supop} was proven in \cite{Gorini:1976uq} and for asking us to provide a simple proof (see \cref{prop:supop.FAF}), and especially to Dr. Jess Riedel for pointing out the connections to, and overlap with \cite{wolf2012quantum,riedel-checkable-2020}, which led us to add \cref{sec:prior} to the final version of this work.

This research was sponsored by the Army Research Office and was
accomplished under Grant Number W911NF-20-1-0075. The views and conclusions 
contained in this document are those of the authors and should not be 
interpreted as representing the official policies, either expressed or implied, 
of the Army Research Office or the U.S. Government. The U.S. Government is 
authorized to reproduce and distribute reprints for Government purposes, notwithstanding any copyright notation herein.

\appendix

\section{Solution of \texorpdfstring{\cref{eq:3}}{(\ref*{eq:3})}
for the coherence vector}
\label{app}

\cref{eq:3} is a linear, first-order, nonhomogeneous differential equation (1ODE). We provide the solution in two parts, first for diagonalizable and invertible $G$, then for general $G$.

\subsection{Solution for diagonalizable and invertible
\texorpdfstring{$G$}{G}}
\label{ss:diag-inv-G}

Let us first assume that $G$ is diagonalizable over $\mathbb{R}^J$ and also invertible.
We look for a solution in the form
\beq
\vec{v}(t) = \vec{v}^{(0)}(t) + \vec{v}^{(\infty)} \ ,
\label{eq:8.7.1v}
\eeq
where $\vec{v}^{(0)}(t)$ is the homogeneous part and $\vec{v}^{(\infty)}$ is the nonhomogeneous (time-independent) part. Let $\vec{x}^{(k)}$ and $\lambda_k$ represent the eigenvectors and (possibly degenerate and complex) eigenvalues of $G$, i.e.,
\beq
G \vec{x}^{(k)} = \lambda_k \vec{x}^{(k)}\ , \qquad k=1,\dots,J \ .
\eeq
It is then straightforward to check by direct differentiation and substitution that 
\bes
\label{eq:8.7.3v}
\begin{align}
    \vec{v}^{(0)}(t) &= \sum_{k=1}^M s_k e^{\lambda_k t}  \vec{x}^{(k)} \\
    \vec{v}^{(\infty)} &= -G^{-1}\vec{c} 
\end{align}
\ees
in the solution in \cref{eq:3}. 
The coefficients $s_k$ are determined by the initial condition $\vec{v}(0)$:
\beq
\vec{v}^{(0)}(0) = \sum_{k=1}^M s_k  \vec{x}^{(k)} = X \vec{s} \ ,\qquad \text{col}_k(X) = \vec{x}^{(k)}\ ,
\eeq
i.e., $X$ is the matrix whose columns are the eigenvectors of $G$. Also, $\vec{v}^{(0)}(0) = \vec{v}(0) - \vec{v}^{(\infty)}$. Thus
\beq
\vec{s} = X^{-1} (\vec{v}(0)+G^{-1}\vec{c}) \ .
\eeq

The eigenvalues can be decomposed as 
$\lambda_k = \real(\lambda_k ) + i \imag(\lambda_k )$. The imaginary part describes a rotation of $\vec{v}$. The real part is constrained by complete positivity and trace preservation to be non-positive, or else the norm of the coherence vector would not be bounded [recall \cref{eq:vnorm}].

\subsection{Solution for general \texorpdfstring{$G$}{G}}
\label{ss:gen-G-sol}

The general case is where $G$ is not diagonalizable over $\mathbb{R}^J$ and may not be invertible. In this case, we can still use a similarity transformation $S$ to transform $G$ into Jordan canonical form:
\beq
G_J = SGS^{-1} = 
\left(
\begin{array}{ccc}
    J_1 &   &   \\
    & \ddots &   \\
    &   & J_q \\
\end{array}
\right)\ ,
\eeq
where the $q$ Jordan blocks have the form
\beq
J_j = 
\left(
\begin{array}{ccccc}
    \mu _j & 1 &   &   &   \\
    & \mu _j & \ddots &   &   \\
    &   & \ddots & \ddots &   \\
    &   &   & \mu _j & 1 \\
    &   &   &   & \mu _j \\
\end{array}
\right) = \m_j I + K_j \ .
\eeq
The $\m_j$'s are the (possibly degenerate, complex) eigenvalues, and $K_j$ are nilpotent matrices: $K_j^{d_j} = 0$, where $d_j$ is the dimension of $J_j$. When all $d_j=1$, $G$ is diagonalizable, and $G_J$ reduces to the diagonalized form of $G$. When one or more of the $d_j>1$, then $G$ is not diagonalizable, meaning that a similarity transformation does not exist that transforms $G$ into a diagonal matrix. The eigenvalues are the solutions of $\det{(G-\mu I)}=0$.

Applying $S$ from the left to \cref{eq:3} yields
\beq
S\dot{\vec{v}} = SGS^{-1}S\vec{v} + S\vec{c} \quad \implies \quad \dot{\vec{w}} = G_J \vec{w} + \vec{c}'\ ,
\label{eq:354}
\eeq
where $\vec{w} = S\vec{v}$ and we defined $\vec{c}' = S\vec{c}$.
We again look for a solution in the form
\beq \label{eq:gen-G-inh}
\vec{w}(t) = \vec{w}^{(0)}(t) + \vec{w}^{(\infty)} \ ,
\eeq
where now $\vec{w}^{(0)}(t)$ is the homogeneous part and $\vec{w}^{(\infty)}$ is the nonhomogeneous (time-independent) part. 

First, let us solve for the homogeneous part.
Since the Jordan blocks are decoupled, the general solution of the homogenous part is 
\beq
\vec{w}^{(0)}(t) = \bigoplus_{j=1}^q \vec{w}_{j}^{(0)}(t) \ ,
\eeq
where we use the direct sum notation to denote that the summands must be stacked into a single-column vector. Each of the $q$ summands satisfies a differential equation of the form
\beq
\dot{\vec{w}}_{j}^{(0)} = J_j {\vec{w}}_{j}^{(0)}\ .
\eeq

\begin{myclaim}
The solution for a general $d_j$ dimensional Jordan block is a vector 
\beq
\vec{w}_j^{(0)} = ({w}_{j,1}^{(0)},\dots,{w}_{j,k}^{(0)},\dots,{w}_{j,d_j}^{(0)})^T
\eeq 
with components:
\bes
\label{eq:jordan-sol}
\begin{align}
{w}_{j,k}^{(0)}(t) &= e^{\m_j t} p_{d_j-k}(t)\\
p_{d_j-k}(t) &= \sum_{n=k}^{d_j} {w}_{j,n}^{(0)}(0) \frac{t^{n-k}}{(n-k)!}\ , \quad k=1,\dots,d_j \ .
\end{align}
\ees
\end{myclaim}
Here $p_{d_j-k}(t)$ denotes a polynomial of degree $d_j-k$. The highest degree of such polynomials is $d_j-1$, of the last component ${w}_{j,d_j}^{(0)}$.
\begin{proof}
From the structure of the $j$'th Jordan block $J_j$, it is clear that
\bes
\begin{align}
\dot{w}_{j,d_j}^{(0)} &= \mu_j {w}_{j,d_j}^{(0)} \\
\dot{w}_{j,k}^{(0)} &= \mu_j {w}_{j,k}^{(0)} + {w}_{j,k+1}^{(0)}\ ,\qquad k\in\{1,\dots d_j-1\} \ ,
\end{align}
\ees
and if we differentiate Eq.~\eqref{eq:jordan-sol} we obtain, first for $k=d_j$: 
\beq
\dot{w}_{j,d_j}^{(0)} = \mu_j e^{\mu_j t} p_0 = \mu_j {w}_{j,d_j}^{(0)}\ ,
\eeq
and second, for $k<d_j$:
\bes
\begin{align}
\dot{w}_{j,k}^{(0)} &= \mu_j {w}_{j,k}^{(0)} + \sum_{n=k+1}^{d_j} {w}_{j,n}^{(0)}(0) \frac{t^{n-k-1}}{(n-k-1)!} \\
&= \mu_j {w}_{j,k}^{(0)} + {w}_{j,k+1}^{(0)}\ ,\qquad k\in\{1,\dots d_j-1\}  \ .
\end{align}
\ees
\end{proof}

Note that we can be certain that for all $d_j>1$, the corresponding $\real(\m_j)< 0$, since a positive or zero real part would violate \cref{eq:vnorm}.

Next, let us solve for the nonhomogeneous part.
As before, it is determined by the solution form \eqref{eq:gen-G-inh} subject to the boundary conditions
\bes\label{eq:gen-G-bc}
\begin{align}
\vec{w}(0)& =  \vec{w}^{(0)}(0) + \vec{w}^{(\infty)} \\
\vec{w}(\infty)& = \vec{w}^{(\infty)}\,.
\end{align}
\ees
On the one hand, $\dot{\vec{w}} = \dot{\vec{w}}^{(0)}$, and on the other hand $\dot{\vec{w}} = G_J (\vec{w}^{(0)}(t) + \vec{w}^{(\infty)}) + \vec{c}' = G_J \vec{w}^{(0)}(t) + G_J\vec{w}^{(\infty)} + \vec{c}'$. Since only the homogeneous part is time-dependent, we must have $\dot{\vec{w}}^{(0)} = G_J \vec{w}^{(0)}(t)$, so that the particular solution satisfies
\beq
G_J\vec{w}^{(\infty)} = -\vec{c}' \ .
\label{eq:winfsol}
\eeq
This linear system has a solution for $\vec{w}^{(\infty)}$ iff $\text{rank}(G_J) = \text{rank}(G_J|\vec{c}')$ (the augmented matrix). 

The solution is unique only if, in addition, $G_J$ is full rank (equal to $J$); otherwise, the solution for $\vec{w}^{(\infty)}$ is under-determined. 
Since $G_J$ is in Jordan form, it is in full-rank row-echelon form unless there are Jordan blocks with $\mu_j=0$.
However, as a consequence of the norm upper bound \cref{eq:vnorm}, any such Jordan blocks must have a dimension $d_j=1$, since otherwise, the polynomials in \cref{eq:jordan-sol} would grow unboundedly.
Therefore, the kernel dimension $\dim(G_J) - \mathrm{rank}(G_J)$ must be the number of Jordan blocks with $\mu_j=0$  and $d_j=1$, and in these blocks, \cref{eq:354} becomes the trivial scalar differential equation $\dot{w}_j = c'_j$, whose solution is $w_j(t) = c'_j t + w_j(0)$. But this means that $c'_j=0$ since otherwise $w_j(t)$ is unbounded. Thus, in each block with $\mu_j=0$ we have $w_j(t) = w_j(0)$.
Such Jordan blocks represent frozen degrees of freedom that do not evolve in time.


\section{Injectivity of
the map \texorpdfstring{$a \mapsto (R,\vec{c})$}{a to (R,c)}}
\label{sec:lemmas12}

Here we provide an alternative proof of injectivity
of the linear map $\mathcal{F}\colon a \mapsto (R,\vec{c})$
introduced in \cref{ss:atoRc}, which repurposes the ideas of \cref{thm:forward.inverse} to focus on injectivity alone. But first, we need a Lemma separate from \cref{thm:forward.inverse}:

\begin{mylemma}
\label{lemma1}
If $\mathcal{F}$ maps matrices $a_1$ and $a_2$ to the same $(R,\vec{c})$, then 
$\mathcal{L}_{a_1} = \mathcal{L}_{a_2}$ where $\mathcal{L}_{a}$ is the 
dissipative Liouvillian of \cref{eq:L_a} acting on arbitrary $d\times d$ operators $X$, not just positive $\rho$ with $\Tr(\rho) = 1$.
\end{mylemma}

\begin{proof}
\cref{prop:4} trivially holds for arbitrary $d\times d$ operators $X$ and vectors $\vec{v}$ with complex-valued components computed using \cref{eq:v_j}, with $\r$ replaced by $X$. It shows that after coordinatization the 
Markovian quantum master equation is equivalent to a linear ordinary differential equation for such vectors, so the solution to 
$\dot{X} = \mathcal{L}_a(X)$ is unique up to initial condition for any $X \in \mc{B}(\mc{H})$.
Suppose $a_1$ and $a_2$ both map to the same $\vec{c}$ and $R$ by \cref{eq:cfroma,eq:R}. Let $X_1(t)$ be the solution to $\dot{X}_1(t) = \mc{L}_{a_1}[X_1(t)]$ and $X_2(t)$ be the solution to $\dot{X}_2(t) = \mc{L}_{a_2}[X_2(t)]$.
Furthermore, let them have the same initial conditions: $X(0) = X_1(0) = X_2(0)$. Let $\vec{v}_1(t)$ and $\vec{v}_2(t)$ be the vectors associated with $X_1(t)$ and $X_2(t)$ via \cref{eq:v_j}. Then we have
\bes
\begin{align}
    \dot{\vec{v}}_1(t) &= R\vec{v}_1(t) + \vec{c}\\
    \dot{\vec{v}}_2(t) &= R\vec{v}_2(t) + \vec{c}\ .
\end{align} 
\ees
By uniqueness of the solution to $\dot{\vec{v}} = R\vec{v}+\vec{c}$ and that $\vec{v}_1(0) = \vec{v}_2(0)$ since $X_1(0) = X_2(0)$, we have $\vec{v}_1(t) = \vec{v}_2(t)$ for all $t$, which implies $X_1(t) = X_2(t)$ for all $t$. Thus $\dot{X} = \mathcal{L}_{a_1}(X)$ and $\dot{X} = \mathcal{L}_{a_2}(X)$ have the same unique solution up to initial conditions. 

For any operator $X(0)$, let $X(t)$ be the unique solution to both $\dot{X} = \mathcal{L}_{a_1}(X)$ and $\dot{X} = \mathcal{L}_{a_2}(X)$ with initial condition $X(0)$. Then $\dot{X}(t)|_0 = \mathcal{L}_{a_1}[X(0)] = \mathcal{L}_{a_2}[X(0)]$ for arbitrary $X(0)$. Thus $\mathcal{L}_{a_1}=\mathcal{L}_{a_2}$.
\end{proof}

\begin{mylemma}
\label{lemma2}
$\mathcal{F}$ is injective (one-to-one).
\end{mylemma}

\begin{proof}
To prove injectivity we need to show that $\mathcal{F}(a_1)=\mathcal{F}(a_2) \implies a_1=a_2$, but since $\mathcal{F}$ is linear it suffices to show that $\mathcal{F}(a)=(0,\vec{0}) \implies a=0$.

Let us take a matrix $a$ such that $\mathcal{F}(a) = (0,\vec{0})$. From \cref{lemma1}:
\begin{equation}
  \label{eq:mLa=0}
  \mathcal{L}_a = \mathcal{L}_0 = 0,
\end{equation}
where $\mathcal{L}_a$ is the dissipative Liouvillian in \cref{eq:LE}. Let $\ket{j}$ be the vector with the $j$-th component equal to $1$ and all other components equal to $0$. Then, from \cref{eq:mLa=0}, $\forall j,k=1,\dots,d$ we have $\mathcal{L}_{a}(\ketb{j}{k}) = 0$. In particular, using \cref{eq:L_a}:
 \bes
 \begin{align}
0 &= \sum_{k=1}^{d} \bra{i}\mathcal{L}_{a}(\ketb{j}{k})\ket{k} \\
&= \,\sum_{{m,n=1}}^{J}\, a_{mn}\big[ \bra{i}F_m\ket{j}\sum_{k=1}^{d} \bra{k}F_n\ket{k} -\frac{1}{2}\bra{i}F_n F_m\ket{j}\sum_{k=1}^{d} \brak{k}{k} \notag \\
&\qquad  - \frac{1}{2}\brak{i}{j}\sum_{k=1}^{d} \bra{k}F_n F_m\ket{k} \big] \\
& = \sum_{m,n=1}^{J} a_{mn}\big[ \bra{i}F_m\ket{j} \Tr(F_n) -\frac{d}{2} \bra{i}F_n F_m\ket{j} \notag \\
&\qquad  - \frac12 \d_{ij} \Tr(F_n F_m) \big]\\
& =  -\frac12 \delta_{ij} \Tr(a) - \frac{d}{2} \sum_{m,n=1}^{J} a_{mn} \bra{i}F_n F_m\ket{j}  \ ,
\label{eq:B2d}
 \end{align}
 \ees
 where we used $\Tr(F_n)=0$ and $\Tr(F_n F_m)=\d_{nm}$ (\cref{def:nice-operator-basis}). For the same reason, $\sum_{i=1}^{d} \bra{i}F_n F_m\ket{i} = \Tr(F_n F_m) = \d_{nm}$, so that after summing \cref{eq:B2d} over $i = j = 1,\ldots,d$ we obtain
\beq
0 = -\frac12 \sum_{i=1}^{d} \Tr(a) - \frac{d}{2} \sum_{m,n=1}^{J} a_{mn} \d_{nm} = -d\Tr(a)\ ,
\eeq 
i.e., $\Tr(a) = 0$. From \cref{eq:B2d} it follows that $\bra{i}A\ket{j} = 0$ $\forall i,j$, where  $A\equiv \sum_{m,n=1}^{J} a_{mn} F_n F_m$, i.e., $A=0$.
Hence, we have shown that $\forall X$:
  \begin{equation}
    \label{eq:La.simple}
    \mathcal{L}_{a}(X) = \sum_{m,n=1}^{J} a_{mn} F_n X F_n\ .
  \end{equation}
We can now express the matrix elements $a_{mn}$ in terms of $\mathcal{L}_{a}$, as follows:
 \bes
 \begin{align}
 a_{mn} &= \sum_{m',n'=1}^{J} a_{m'n'} \Tr(F_n F_{n'})\Tr(F_{m'} F_{m}) \\
 & = \sum_{j,k=1}^{d} \bra{j}F_n\big[\sum_{m',n'=1}^{J}a_{m'n'} F_{n'}\ketb{j}{k}F_{m'}\big]F_m\ket{k}\\
 & = \sum_{j,k=1}^{J} \bra{j}F_n  \mathcal{L}_{a}(\ketb{j}{k}) F_m\ket{k} = 0\ ,
 \end{align}
 \ees
since $\mathcal{L}_{a} = 0$. That is, $a = 0$.
\end{proof}

\section{Facts from convex geometry}
\label{sec:convex}
In this section, we summarize the facts from convex geometry
needed for this work.

Let $M_{{+}}(J,\mathbb{C})$ be the set
of positive semidefinite $J\times J$ matrices with complex coefficients:
\begin{equation}
  \label{eq:Mplus.def}
  M_{{+}}(J,\mathbb{C}) = \left\{a \in M_{\textrm{sa}}(J, \mathbb{C}):
  \forall b \in \mathbb{C}^J\; b^\dag a b \geq 0 \right\}.
\end{equation}
This equation describes $M_{{+}}(J,\mathbb{C})$ as the intersection
of infinitely many half-spaces in $M_{\textrm{sa}}(J,\mathbb{C})$,
whose boundaries all pass through
the origin $a=0$. Therefore, $M_{{+}}(J,\mathbb{C})$ is a closed convex
cone in $M_{\textrm{sa}}(J,\mathbb{C})$.

Let
\begin{equation}
  \label{eq:Mplus1.def}
  M_{{+},1}(J,\mathbb{C}) = \{a \in M_{{+}}(J,\mathbb{C})\colon\; \Tr(a) = 1\},
\end{equation}
i.e., the set of unit-trace, positive-semidefinite matrices.

\begin{mylemma}
  \label{lm:Mplus.1}
  $M_{{+}}(J,\mathbb{C})$ is a closed convex cone with a compact base, which 
  can be chosen to be $M_{{+},1}(J,\mathbb{C})$.
\end{mylemma}
\begin{proof}
  As noted above, $M_{{+}}(J,\mathbb{C})$ is a closed convex cone.
  By definition, $M_{{+},1}(J,\mathbb{C})$ is an intersection of that
  cone and a hyperplane not passing through the origin.
  Thus, it remains to prove that
  \begin{enumerate}
    \item any of the rays in $M_{{+}}(J,\mathbb{C})$ passes through
    $M_{{+},1}(J,\mathbb{C})$, i.e., every non-zero element of $M_{{+}}(J,\mathbb{C})$ 
    is proportional to an element of $M_{{+},1}(J,\mathbb{C})$ with a positive 
    coefficient;
    \item $M_{{+},1}(J,\mathbb{C})$ is compact.
  \end{enumerate}
  To prove the first statement, let $a \neq 0, a \in M_{{+}}(J,\mathbb{C})$. 
  Since $a$ is Hermitian
  it can be diagonalized. \cref{eq:Mplus.def} implies that all its eigenvalues
  are non-negative, hence $\Tr(a) \geq 0$. In fact, at least one of the
  eigenvalues has to be positive, otherwise $a = 0$, hence $\Tr(a) > 0$.
  Therefore $a = \Tr(a) (a / \Tr(a))$ where the coefficient $\Tr(a) > 0$ and
  the element $a / \Tr(a) \in M_{{+},1}(J,\mathbb{C})$.

  Since $M_{{+},1}(J,\mathbb{C})$ is an intersection of closed sets,
  it is closed. Thus, it remains to prove it is bounded, e.g.:
  \begin{equation}
    \label{eq:Mplus1.a.bounded}
    \forall a \in M_{{+},1}(J,\mathbb{C}) \quad
    \sum_{i,j=1}^{J} \abs{a_{ij}}^2 \leq 1\ .
  \end{equation}
  Let $\{\lambda_i\}_{i=1}^{J}$ be the eigenvalues of such an $a$. Then
  $\sum_{i=1}^J \lambda_i = 1$ and $0 \leq \lambda_i \leq 1$. Therefore, indeed,
  \begin{equation}
    \sum_{i,j=1}^{J} \abs{a_{ij}}^2 = \Tr(a^\dag a) = \sum_{i=1}^J \lambda_i^2
    \leq \sum_{i=1}^J \lambda_i = 1\ .
  \end{equation}
\end{proof}

\begin{mylemma}
  \label{lm:Mplus.2}
  The extreme rays of $M_{{+}}(J,\mathbb{C})$ are the rays generated by rank 1 
  matrices.
\end{mylemma}
Note: this is \cite[Exercise 7.5.15]{Horn:book}. Here we provide a proof for
completeness.
\begin{proof}
  Let $\{\alpha a: \alpha \in \mathbb{R}_+\}$ be an extreme ray in
  $M_{{+}}(J,\mathbb{C})$. Here $a \in M_{{+}}(J,\mathbb{C}) \setminus \{0\}$.
  Since $a$ is Hermitian, we can diagonalize it by a unitary,
  hence write it in the form
  \begin{equation}
    \label{eq:Mplus.abb}
    a = \sum_{i=1}^k b_i b_i^\dag
  \end{equation}
  for some $b_i \in \mathbb{C}^J$ orthogonal to each other, where
  $k = \rank(a)$. Since $a \neq 0$, $k \geq 1$. If $k \geq 2$ then
  \cref{eq:Mplus.abb} represents $a$ as a sum of non-proportional
  elements of $M_{{+}}(J,\mathbb{C})$, contradicting the assumption
  that $a$ lies on an extreme ray in $M_{{+}}(J,\mathbb{C})$. Thus, $k=1$.

  Vice-versa, if $\rank(a) = 1$, $a\in M_{{+}}(J,\mathbb{C})$, we want to
  prove that $a$ lies on an extreme ray of $M_{{+}}(J,\mathbb{C})$.
  Assume to the contrary that
  \begin{equation}
    \label{eq:Mplus.adecomp}
    a = a_1 + a_2
  \end{equation}
  where $a_1,a_2 \in M_{{+}}(J,\mathbb{C})$ and $a_1$ is not proportional to 
  $a$. We can write $a = \alpha bb^\dag$ for some
  $\alpha \in \mathbb{R}_+$, $b\in \mathbb{C}^J$ with 
  $\norm{b}=1$, $\alpha > 0$. Then for $i=1,2$, we can decompose
  \begin{equation}
    a_i / \alpha = \beta_i bb^\dag + \tilde{a}_i
  \end{equation}
  where $\beta_i = b^\dag a_i b / \alpha$, $b^\dag\tilde{a}_ib = 0$.
  \cref{eq:Mplus.adecomp} implies that $\beta_1 + \beta_2 = 1$,
  $\tilde{a}_1 + \tilde{a}_2 = 0$. By our assumption $\tilde{a}_1 \neq 0$ 
  (otherwise $a_1$ would be proportional to $a$).
  Therefore, there is a vector $c$ s.t. $c^\dag\tilde{a}_1c \neq 0$.
  Let $c = \gamma b + c_{\perp}$, where $b^\dag c_{\perp} = 0$.
  Note that
  $c_\perp^\dag \tilde{a}_1 c_\perp = c_\perp^\dag a_1 c_\perp / \alpha$,
  hence $c_\perp^\dag \tilde{a}_1 c_\perp \geq 0$. Similarly,
  $c_\perp^\dag \tilde{a}_1 c_\perp = -c_\perp^\dag a_2 c_\perp / \alpha$,
  hence $c_\perp^\dag \tilde{a}_1 c_\perp \leq 0$. Thus,
  $c_\perp^\dag \tilde{a}_1 c_\perp = 0$.

  For any $\delta \in \mathbb{C}$ we have
  \begin{equation}
    \label{eq:Mplus.quadratic}
    (b + \delta c_{\perp})^\dag a_1 (b + \delta c_{\perp}) \geq 0\ ,
  \end{equation}
and we have just shown that the quadratic term in this expression is zero. Since 
  the inequality \cref{eq:Mplus.quadratic} has to hold for all $\delta$, the 
  linear term is zero too, i.e., $b^\dag a_1 c_{\perp} = 0$.
  Thus, $c^\dag \tilde{a}_1 c = 0$, contradicting the assumption.
\end{proof}

\begin{mylemma}
  \label{lm:Mplus.3}
  $M_{{+}}(J,\mathbb{C})$ is the convex hull of its extreme rays.
\end{mylemma}
Here we present two alternative ways to see this.
\begin{proof}[Proof 1]
  Any positive semidefinite matrix can be diagonalized. Such diagonalization
  results \cite[Theorem 7.5.2]{Horn:book} in a decomposition of the matrix into 
  rank-$1$ positive semidefinite matrices [as in \cref{eq:Mplus.abb}], which -- according to \cref{lm:Mplus.2} --
  lie on the extreme rays of $M_{{+}}(J,\mathbb{C})$.
\end{proof}
\begin{proof}[Proof 2]
  In the context of finite dimensional vector spaces,
  Krein-Milman (also known as Minkowski's theorem, see, e.g.,
  \cite[Corollary 18.5.1]{Rockafellar:book}) states that
  any compact convex set is the convex hull of its extreme points.
  As a corollary any convex cone with a compact base is the 
  convex hull of its extreme rays. As follows from \cref{lm:Mplus.1},
  $M_{{+}}(J,\mathbb{C})$ is such a cone.
\end{proof}

\section{Representation of a superoperator using a nice operator basis}
\label{sec:supop.FAF}

\begin{myproposition}
  \label{prop:supop.FAF}
  Any superoperator 
  $\mc{E}\in \mathcal{B}[\mathcal{B}(\mc{H})]$ can always be represented as
  \begin{equation}
    \label{eq:Ph-gen-supop2}
    \mc{E}(A) = \sum_{{i,j=0}}^{J} c_{ij} F_i \bullet F_{j}\ .
  \end{equation}
\end{myproposition}

This is \cite[Lemma 2.2]{Gorini:1976uq}.

It can also be seen directly by noting that any linear operator can be
represented by a matrix and, thus,
$\mathcal{E}(A)_{kn} = \sum_{l,m=1}^{d} \mathcal{E}_{klmn} A_{lm}$ for some
tensor $\mathcal{E}_{klmn}$. The tensor $\mathcal{E}_{klmn}$ can be seen as a
matrix in 2 ways: first, as a function of indices $k$ and $l$, and second, as a
function of indices $m$ and $n$. Applying ``coordinatization"
[\cref{eq:coordinatization1,eq:coordinatization2}] twice to $\mathcal{E}_{klmn}$ we get $c_{ij}$.

\begin{proof}
  Let $\mathcal{E}$ be any superoperator in $\mathcal{B}[\mathcal{B}(\mc{H})]$.
  Denoting $\mathcal{E}_{klmn} = \mathcal{E}(\ketb{l}{m})_{kn}$ and
  applying \cref{eq:coordinatization4} to 
  $\delta_{kk'} \delta_{ll'}$ and $\delta_{mm'} \delta_{nn'}$ for any $A \in \mathcal{B}(\mathcal{H})$ we get
  \bes
  \begin{align}
    &\mathcal{E}(A)_{kn} = \sum_{l,m=1}^{d}\mathcal{E}_{klmn}A_{lm} \\
    &\quad = \sum_{k',l',m',n',l,m=1}^{d}
      \delta_{kk'} \delta_{ll'} \delta_{mm'} \delta_{nn'}
      \mathcal{E}_{k'l'm'n'} A_{lm}  \\
    &\quad = \sum_{k',l',m',n',l,m=1}^{d}
     \sum_{i,j=0}^{J} F_{i\,kl} F_{i\,l'k'} F_{j\,mn} F_{j\,n'm'}
      \mathcal{E}_{k'l'm'n'} A_{lm} \\
    &\quad =  \sum_{i,j=0}^{J} c_{ij} (F_i A F_j)_{kn},
  \end{align}
  \ees
  where
  \begin{equation}
    c_{ij} = \sum_{k',l',m',n'=1}^{d} F_{i\,l'k'} F_{j\,n'm'} \mathcal{E}_{k'l'm'n'}.
  \end{equation}
\end{proof}

\section{Computation of the constant \texorpdfstring{$\theta^{\text{GinOE}}$}{θᴳⁱⁿᴼᴱ}}
\label{app:theta}

The constant $\theta^{\text{GinOE}}$ in \cref{eq:P-GinOE} can be computed as
\begin{equation}
    \theta^{\text{GinOE}} = \inf_{\rho\in \mathcal{S}_{{+}}} E(\rho) - \inf_{\rho\in \mathcal{S}} E(\rho),
\end{equation}
where
\bes
\begin{align}
    E(\rho) &= \frac12 \int_{\mathbb{C}} \left(
        \abs{z}^2 - \int_{\mathbb{C}} \ln\abs{z-w}\rho(w) d^2 w
      \right)\rho(z)d^2z,\\
    \mathcal{S} &= \Bigl\{\rho \in C_0(\mathbb{C} \to \mathbb{R}_{+}):
      \int_{\mathbb{C}}\rho(z) d^2z=1,\notag\\
      &\qquad \forall z\;\rho(z)=\rho(z^*)\Bigr\},\\
    \mathcal{S}_{{+}} &= \left\{\rho \in \mathcal{S}:
      \forall z\; \real(z) < 0 \Rightarrow \rho(z) = 0\right\},\\
    d^2 z &= d\real(z) d \imag(z) = \frac{i}{2} dzdz^*,
\end{align}
\ees
where $C_0(X\to Y)$ is the set of continuous functions from $X$ to $Y$ 
with compact support. One can then compute that
$\inf_{\rho\in \mathcal{S}} E(\rho) = 5/8 - \ln(2)/4$ and is achieved 
near $\rho(z) = \frac{1}{2\pi} \mathbb{I}(\abs{z}^2 < 2)$, where $\mathbb{I}$ is the indicator function. One can then try to argue that the $\inf_{\rho\in\mathcal{S}_{{+}}}E(\rho)$
is achieved for $\rho$ of the following form:
\begin{equation}
  \label{eq:rho.anzats}
  \rho(x+iy) = g(y)\delta(x) + \frac{1}{2\pi}\mathbb{I}(0\leq x \leq f(y)).
\end{equation}
Note that while the function \cref{eq:rho.anzats} is not continuous, it can be approximated by
a continuous function in $\mathcal{S}_{{+}}$ as soon as $f$ and $g$ are non-negative continuous functions with compact support and
\begin{equation}
    \int_{\mathbb{R}}\left(g(y) + \frac{f(y)}{2\pi}\right)dy=1.
\end{equation}
For such $\rho$ we have
\begin{multline}
  E(\rho) = \int_{\mathbb{R}} dy_1 \Biggl(\frac{f(y_1)^3}{12\pi} + \frac{y_1^2 f(y_1)}{4\pi} + g(y_1)\frac{y_1^2}{2}-\\
  \int_{\mathbb{R}} dy_2 \biggl(g(y_1)g(y_2)\ln\abs{y_1-y_2}/2 +\\
  \frac{1}{2\pi}g(y_2)I_1(f(y_1),\abs{y_2-y_1})+\\
  \frac{1}{8\pi^2} I_2(f(y_1),f(y_2),y_2-y_1) \biggr)
  \Biggr),
\end{multline}
where
\bes
\begin{align}
    I_1(x,y) &= \int_{0}^{x} dx_1 \ln\abs{x_1+iy},\\
    I_2(x_1,x_2,y) &= \int_{0}^{x_1} dx_3 \int_{0}^{x_2} dx_4 \ln\abs{x_2-x_4+iy}.
\end{align}
\ees
One can then solve this optimization problem numerically (or attempt to find an analytical solution) to find $\theta^{\text{GinOE}}$.

\end{document}